\documentclass[11pt]{article}

\input glyphtounicode
\usepackage[T1]{fontenc}
\usepackage[utf8]{inputenc}

\usepackage[english]{babel}

\usepackage{mathtools}
\usepackage{amsthm}

\usepackage{enumitem}

\usepackage{xcolor}
\definecolor{hypercolor}{rgb}{0,0.2,0.7}

\usepackage{tikz}

\usepackage[numbers,sort&compress]{natbib}
\usepackage{hyperref}
\hypersetup{
  final,
  unicode,
  breaklinks,
  colorlinks,
  linkcolor=hypercolor,
  urlcolor=hypercolor,
  citecolor=hypercolor,
  bookmarksnumbered,
  bookmarksdepth=2
}

\allowdisplaybreaks

\usepackage[labelfont=bf]{caption}

\numberwithin{equation}{section}

\newif\iffancyfont%
\fancyfontfalse%

\IfFileExists{MinionPro.sty}{\IfFileExists{MyriadPro.sty}{\fancyfonttrue}{}}{}

\iffancyfont
  \usepackage[final]{microtype}
  \usepackage[lf,medfamily,italicgreek,openg]{MinionPro}
  \usepackage[cal=dutchcal]{mathalfa}

  \renewcommand{\dotsb}{{\mathinner{\cdotp\cdotp\cdotp}}}
  \renewcommand{\dotsm}{{\mathinner{\cdotp\cdotp\cdotp}}}

  \DeclareRobustCommand{\defn}{\mathrel{{\vdotdot}{\equal}}}

  \usepackage[toc,eqno,bib,enum]{tabfigures}
  \newtagform{tabular}[\figureversion{tabular}]{(}{)}
  \usetagform{tabular}

  \usepackage{dsfont}
  \newcommand{\one}{\mathds{1}}

  \DeclarePairedDelimiter{\parap}{\lsem}{\rsem}
\else
  \usepackage{lmodern}
  \usepackage{amssymb,amsfonts}

  \let\uppi\pi

  \newcommand{\defn}{\coloneqq}

  \newcommand{\one}{{{\mathchoice {\rm 1\mskip-4mu l} {\rm 1\mskip-4mu l} {\rm 1\mskip-4.5mu l} {\rm 1\mskip-5mu l}}}}

  \usepackage{mathabx}
  \DeclarePairedDelimiter{\parap}{\ldbrack}{\rdbrack}
\fi

\newcounter{smallarabics}
\newenvironment{arabicenumerate}
{\begin{list}{{\normalfont\textrm{(\arabic{smallarabics})}}}
  {\usecounter{smallarabics}\setlength{\itemindent}{0cm}
   \setlength{\leftmargin}{5ex}\setlength{\labelwidth}{4ex}
   \setlength{\topsep}{0.75\parsep}\setlength{\partopsep}{0ex}
   \setlength{\itemsep}{0ex}}}
{\end{list}}

\newcounter{smallroman}

\def\ben#1\een{\begin{arabicenumerate}#1\end{arabicenumerate}}

\newcommand\Diag{\mathrm{Diag}}

\newcommand\p{{\rm \partial}}

\newcommand\T{{\rm T}}

\newcommand\rd{\mathrm{d}}

\newcommand\starprod{\star}

\newcommand\tr{{\rm tr}}

\renewcommand\i{{\rm i}}

\def\otimesal{\mathop{\hbox{\raise 1.5 ex
  \hbox{$\scriptscriptstyle\rm al$}
\kern -0.92 em \hbox{$\otimes$}}}}
\def\oplusal{\mathop{\hbox{\raise 1.5 ex
  \hbox{$\scriptscriptstyle\rm al$}
\kern -0.92 em \hbox{$\oplus$}}}}
\def\Gammal{\hbox{\raise 1.68 ex
\hbox{$\scriptscriptstyle\rm al$}\kern -0.50 em $\Gamma$}}

\newcommand\rr{{\mathbb R}}

\newcommand\cD{{\mathcal D}}

\newcommand\cS{{\mathcal S}}

\newcommand\cX{{\mathcal X}}

\renewcommand{\Re}{\operatorname{Re}}

\DeclareMathOperator{\supp}{supp}

\DeclareMathOperator{\Op}{Op}

\DeclareMathOperator{\Tr}{Tr}

\DeclarePairedDelimiter{\abs}{\lvert}{\rvert}

\DeclarePairedDelimiter{\jnorm}{\langle}{\rangle}

\renewcommand\bar{\overline}

\newcommand\e{{\rm e}}

\renewcommand\d{\mathop{}\!\mathrm{d}}
\renewcommand\i{{\mathrm i}}

\newcommand{\conj}[1]{\overline{#1}}
\newcommand{\nnabla}{\nabla}
\newcommand{\ppartial}{{\partial_p}}
\newcommand{\rlapprime}{\mathrlap{{}'}}

\newcommand{\bs}{\boldsymbol}

\def\beq#1\eeq{\begin{equation}#1\end{equation}}
\def\bep#1\eep{\begin{proposition}#1\end{proposition}}
\def\ber#1\eer{\begin{remark}#1\end{remark}}
\def\bel#1\eel{\begin{lemma}#1\end{lemma}}
\def\bet#1\eet{\begin{theorem}#1\end{theorem}}
\def\bex#1\eex{\begin{example}#1\end{example}}
\def\bed#1\eet{\begin{definition}#1\end{definition}}
\def\bea#1\eet{\begin{assumption}#1\end{assumption}}
\def\bec#1\eec{\begin{corollary}#1\end{corollary}}

\newtheorem{theorem}{Theorem}[section]
\newtheorem{proposition}[theorem]{Proposition}
\newtheorem{lemma}[theorem]{Lemma}
\newtheorem{definition}[theorem]{Definition}
\newtheorem{corollary}[theorem]{Corollary}
\newtheorem{remark}[theorem]{Remark}
\newtheorem{example}[theorem]{Example}
\newtheorem{assumption}{Assumption}[section]

\begin{document}

\hypersetup{
  pdfauthor={Jan Dereziński, Adam Latosiński, Daniel Siemssen},
  pdftitle={Pseudodifferential Weyl Calculus on (Pseudo-)Riemannian Manifolds}
}

\title{Pseudodifferential Weyl Calculus on (Pseudo-)Riemannian Manifolds}

\author{
  Jan Derezi\'{n}ski,
  Adam Latosi\'{n}ski
  \\[.5em]
  Department of Mathematical Methods in Physics, Faculty of Physics\\
  University of Warsaw\\
  Pasteura 5, 02-093 Warszawa, Poland\\
  E-mail:
  \href{mailto:jan.derezinski@fuw.edu.pl}{jan.derezinski@fuw.edu.pl},
  \href{mailto:adam.latosinski@fuw.edu.pl}{adam.latosinski@fuw.edu.pl}\\
  \\[.5em]
  Daniel Siemssen\thanks{corresponding author}
  \\[.5em]
  Department of Mathematics\\
  University of York\\
  Heslington, York YO10 5DD, United Kingdom\\
  E-mail: \href{mailto:daniel.siemssen@york.ac.uk}{daniel.siemssen@york.ac.uk}
}

\maketitle

\begin{abstract}
  One can argue that on flat space $\mathbb{R}^d$ the Weyl quantization is the most natural choice and that it has the best properties (e.g., symplectic covariance, real symbols correspond to Hermitian operators).
  On a generic manifold, there is no distinguished quantization, and a quantization is typically defined chart-wise.
  Here we introduce a quantization that, we believe, has the best properties for studying natural operators on pseudo-Riemannian manifolds.
  It is a generalization of the Weyl quantization -- we call it the \emph{balanced geodesic Weyl quantization}.
  Among other things, we prove that it maps square integrable symbols to Hilbert-Schmidt operators, and that even (resp.\ odd) polynomials are mapped to even (resp.\ odd) differential operators.
  We also present a formula for the corresponding star product and give its asymptotic expansion up to the 4th order in Planck's constant.
\end{abstract}

\section{Introduction}

\emph{Quantization} means representing operators by their \emph{classical symbols}, that is, functions on phase space.
This concept first appeared in quantum physics in the early 20th century.
Only later this idea was adopted in pure mathematics.
Mathematicians, in particular, introduced the so-called \emph{pseudodifferential operators}, defined as quantizations of certain symbol classes.
\emph{Pseudodifferential calculus} is nowadays a major tool in the study of partial differential equations.

Pseudodifferential calculus is useful both on~$\rr^d$ and on manifolds.
In the case of~$\rr^d$, one can use its special structure to define several distinguished quantizations.
One of them is the \emph{Weyl} or \emph{Weyl--Wigner quantization}.
One can argue that it is the most natural quantization and that it has the best properties.
Let us list some of its advantages:
\begin{enumerate}
  \item The Weyl quantization is invariant with respect to the group of linear symplectic transformations of the phase space $\T^*\rr^d=\rr^d\oplus\rr^d$.
  \item Real symbols correspond to Hermitian operators.
  \item Error terms in various formulas have a smaller order in~$\hbar$ if we use the Weyl quantization than if we use other kinds of quantizations.
  \item The Weyl quantization of even, resp.\ odd polynomials is an even, resp.\ odd differential operator.
  \item The Weyl quantization is proportional to a unitary operator if symbols are equipped with the natural scalar product and operators are equipped with the Hilbert--Schmidt scalar product.
\end{enumerate}

In order to define a quantization on a generic manifold $M$, one typically covers it by local charts, and then uses the formalism from the flat case within each chart.
This construction obviously depends on coordinates and thus there is no distinguished quantization on a generic manifold.
In particular, considering the natural symplectic structure of the cotangent bundle~$\T^*M$, quantizations are symplectically invariant only on the level of the so-called \emph{principal symbol}.
This is not a problem for many applications of pseudodifferential calculus, which are quite rough and qualitative.
In such applications, it is even not very important which kind of quantization one uses.
In more quantitative applications, it is more important to choose a good quantization, which usually means a version of the Weyl quantization.

Suppose in addition that $M$ is pseudo-Riemannian.
One can try to use its structure to define a quantization that depends only on the geometry of $M$.
In order to discuss various possibilities, let us first assume that $M$ is \emph{geodesically simple}, that is, each pair of points $x,y \in M$ can be joined by exactly one geodesic.

Here is one possible proposal of a geometric\footnote{Geometric in the sense of relying on the pseudo-Riemannian geometry of the manifold.} generalization of the Weyl quantization: given a function $\T^*M\ni(z,p)\mapsto b(z,p)$, we define its \emph{naive geodesic Weyl quantization} to be the operator with the integral kernel, which for $x,y\in M$ is defined as
\begin{equation}\label{eq:symbol-to-kernel-}
  \Op_{\text{naive}}(b)(x,y) \defn \int_{\T_z^*M} b(z,p) \e^{-\i u \cdot p} \frac{\d p}{(2\uppi)^d},
\end{equation}
where $z$ is the middle point of the geodesic joining $x$ and $y$, $u$ is the tangent vector at $z$ such that $x=\exp_z(-\frac12 u)$ and $y=\exp_z(\frac12 u)$, and we integrate over the variable $p$ in the cotangent space $\T_z^*M$.
The quantization $\Op_{\text{naive}}$ is geometric (does not depend on coordinates) and it reduces to the Weyl quantization if $M=\rr^d$.
However, there exist better definitions, as we argue below.

In our paper we propose to multiply the right-hand side of~\eqref{eq:symbol-to-kernel-} by
\begin{equation}\label{pre1}
  \frac{\abs{g(x)}^\frac14 \abs{g(y)}^\frac14} {\abs{g(z)}^{\frac12}},
\end{equation}
the product of the appropriate roots of the determinants of the metric at $x$, $y$ and $z$.
With this factor we obtain an integral kernel which is a half-density in both $x$ and $y$.
Then we multiply it by the biscalar
\begin{equation}\label{pre2}
  \Delta(x,y)^\frac12,
\end{equation}
the square root of the so-called \emph{Van Vleck--Morette determinant}.\footnote{%
  The (complicated) story of the Van Vleck--Morette determinant can be found in an interesting article~\cite{ChoSt}, where it is argued that its correct name should be the Morette--van Hove determinant.
  We, however, use the name that seems to be established in the differential geometry community.
}
We obtain
\begin{equation}\label{eq:symbol-to-kernel+}
  \Op(b)(x,y) \defn \frac{\Delta(x,y)^\frac12 \abs{g(x)}^\frac14 \abs{g(y)}^\frac14} {\abs{g(z)}^{\frac12}} \int_{\T_z^*M} b(z,p) \e^{-\i u \cdot p} \frac{\d p}{(2\uppi)^d},
\end{equation}
which we call the \emph{balanced geodesic Weyl quantization of the symbol $b$}.
This quantization belongs to a family of quantizations considered in~\cite{fulling:aa}, although for scalar functions instead of half-densities.

Note that the Van Vleck--Morette determinant is a biscalar, and therefore $\Op(b)(x,y)$ is a half-density in both~$x$ and~$y$, which is appropriate for the integral kernel of an operator acting on smooth, compactly supported half-densities on~$M$.
Moreover, the balanced geodesic Weyl quantization has a remarkable property: it satisfies
\begin{equation}\label{paza}
  \Tr\bigl(\Op(a)^*\Op(b)\bigr) = \int_{\T^*M} \bar{a(z,p)} b(z,p) \d z\, \frac{\d p}{(2\uppi)^d}.
\end{equation}
This is the analog of property (5) from the list of advantages of the Weyl quantization in the flat case.
Actually, the balanced geodesic Weyl quantization has all the advantages from that list except for (1).

It is easy to see that a quantization which reduces to the Weyl quantization in the flat case and satisfies~\eqref{paza} is essentially unique and is given by~\eqref{eq:symbol-to-kernel+}.

In general pseudo-Riemannian manifolds, there can be no or many geodesics joining pairs of points.
However, there always exists a certain neighborhood $\Omega$ of the diagonal such that each pair $(x,y)\in\Omega$ is joined by a distinguished geodesic.
Therefore, on a general pseudo-Riemannian manifold the definition~\eqref{eq:symbol-to-kernel+} (and also~\eqref{eq:symbol-to-kernel-}) makes sense only inside $\Omega$.
To make it global, one can insert a smooth cutoff supported in $\Omega$, equal to $1$ in a neighborhood of the diagonal.
This is not a serious drawback, since in practice pseudodifferential calculus is mostly used to study properties of operators close to the diagonal.
Note also that this cutoff does not affect $\Op(a)$ if $a$ is a polynomial in the momenta, since the kernel of $\Op(a)$ is then supported inside the diagonal.

We are convinced that our quantization is a good tool for studying natural operators on~$M$, such as the Laplace--Beltrami operator~$\Delta$ in the Riemannian case and the d'Alembert (wave) operator~$\Box$ in the Lorentzian case.
In our future papers, we plan to describe applications of this quantization to computing singularities at the diagonal of the kernel of the heat semigroup $\e^{\tau\Delta}$, Green's operator $(\Delta+m^2)^{-1}$, the proper time dynamics $\e^{\i\tau\Box}$, the Feynman propagator $(\Box-\i0)^{-1}$, etc.

Besides the geodesic Weyl quantization, we introduce also a whole family of quantizations parametrized by $\tau\in[0,1]$.
The Weyl quantization corresponds to $\tau=\frac12$.
All of them satisfy obvious analogs of the identity~\eqref{paza}.
One can argue that the cases $\tau=0$ and $\tau=1$ are also of practical interest.
However, the case $\tau=\frac12$ typically leads to the most symmetric algebraic expressions.
As mentioned above, the main application of our pseudodifferential calculus is to obtain asymptotic expansions of integral kernels $B(x,y)$ around the diagonal $x=y$.
The geodesic Weyl quantization gives such expansions around the middle point of the geodesic joining $x$ and~$y$.
If we use the $\tau=0$ quantization, then the expansion uses $x$ as the central point, which is less symmetric and involves less cancellations.

For many authors, the philosophy of quantization is quite different from ours.
Some authors study quantization as an end in itself.
Others are interested only in applications which are quite robust and to a large extent insensitive to the choice of the quantization.
Such applications include propagation of singularities, elliptic regularity, computation of the index of various operators.
Other applications, such as spectral asymptotics, are more demanding in the choice of quantization.
Our aim is to have an efficient tool for computing the asymptotics of various operators, giving an expansion which is as simple as possible.
We have already checked this when computing the Feynman propagator on a Lorentzian manifold.
We started from the naive geodesic Weyl quantization, and we discovered empirically that inserting the prefactors \eqref{pre1} and \eqref{pre2} decreases substantially proliferation of various error terms.

There is another argument why the balanced Weyl quantization is superior to the naive one:
If one computes asymptotics of heat kernels or Feynman propagators using the traditional methods, without the momentum variables, as, e.g., in \cite{avramidi:ab,avramidi:heat}, then the
prefactors \eqref{pre1} and \eqref{pre2} appear.
So these prefactors simplify the expressions one looks for.

Our original motivation for introducing the balanced geodesic Weyl quantization comes from quantum field theory on curved spacetimes.
One would like to define renormalized Wick powers of fields and its derivatives using a scheme that depends only on the local geometry.
On a (flat) Minkowski space, renormalization is usually done in the momentum representation.
There are in fact several (essentially equivalent) schemes for renormalization that use momentum representation.
It is usually stated that on curved spaces the momentum representation is not available, and one has to use the position representation, which is much more complicated.
The main tool for renormalization is then the asymptotics around the diagonal of the Feynman propagator, or what is equivalent, of the so-called Hadamard state defining the two point function.
The use of a quantization allows us to use the momentum representation for renormalization on curved spacetimes.
In our opinion, the balanced geodesic Weyl quantization will lead to the simplest computations in this context.

Quantization and pseudodifferential calculus is an old subject with an interesting history and large literature.
The Weyl quantization was first proposed by Weyl in 1927 \cite{Weyl}.
Wigner was the first who considered its inverse, called sometimes the \emph{Wigner function} or the \emph{Wigner transform}~\cite{Wig}.
The \emph{star product} (the product of two operators on the level of symbols) and the identity~\eqref{paza} were first described by Moyal~\cite{Moy}.

Pseudodifferential calculus became popular after the paper by Kohn{--\penalty0\hskip0pt\relax}Nirenberg~\cite{KN}.
Kohn--Nirenberg used the $\tau=1$ quantization.
The usefulness of the Weyl quantization in PDE's was stressed by H\"ormander~\cite{hormander:aa,hormander:ae}.
Weyl quantization is also discussed, e.g., in~\cite{shubin:ad,derezinski:ac}.

Several quantizations on manifolds were considered before in the literature.
Working either with the Levi-Civita connection on a Riemannian manifold, with an arbitrary connection, or with the so-called linearization introduced in~\cite{bokobza-haggiag:aa,bokobza-haggiag:ab}, most works attempt to generalize the $\tau=0$ or $\tau=1$ quantization, see e.g.~\cite{bokobza-haggiag:aa,bokobza-haggiag:ab,fulling:ab,widom:aa}, the more recent works~\cite{pflaum:aa,pflaum:ab,sharafutdinov:aa,sharafutdinov:ab,safarov,bordemann} and the references therein.
An extensive introduction to early results on this topic can be found in~\cite{fulling:ab}.
Besides addressing the intricacies of defining a quantization on a manifold (e.g., caustics of geodesics), some of these works discuss symbol classes, the star product, heat kernel and resolvent computations (although only to low order).

A generalization of the Weyl quantization to manifolds with a connection was advocated by Safarov~\cite{safarov,demuth:aa} and is essentially equivalent to what we call the naive geodesic Weyl quantization.
Similar definitions can be found also in other places in the literature, e.g.~\cite{Nest}.
More recently, Levy~\cite{levy:aa} considered a similar generalization for manifolds with linearization.
Like in our manuscript, these papers consider the whole class of $\tau$-quantizations on manifolds.
They, however, do not use the geometric factor involving the Van-Vleck--Morette determinant, which as we argued above, improves the properties of a quantization.

Ideas very close to those of our paper were discussed by Fulling~\cite{fulling:aa}, who analyzed the effect of various powers of the Van Vleck--Morette determinant on the quantization.
Fulling, in particular, remarked that the square root of this determinant may be viewed as the distinguished choice.
He also noticed the term $\frac16R$, which appears when one tries to define the Laplace--Beltrami or d'Alembert operator using the Weyl-type quantization involving the Van Vleck--Morette determinant.
There is one difference between our approach and Fulling's: he used scalars, whereas we use half-densities -- of course, this is a minor difference, since it is easy to pass from one framework to the other.
We go much further than Fulling in the analysis of the balanced geodesic Weyl quantization:
we prove the identity~\eqref{paza}, which we view as a key advantage of this quantization, and we analyze the corresponding star product.

Güntürk in his PhD thesis \cite{gunturk:aa} attempted to develop a Weyl calculus on Riemannian manifolds, including the star product.
His approach, however, involved some constructions depending on coordinates, hence it was not fully geometric.

As we stressed above, the balanced geodesic Weyl quantization is geometric -- by this we mean that it essentially depends only on the geometry of a pseudo-Riemannian manifold $M$.
However,  one should mention that the name \emph{geometric quantization} has already a well-established meaning, which involves a somewhat different setting.
The usual starting point of geometric quantization is a symplectic or, more generally, a Poisson manifold.
Then one tries to define a noncommutative associative algebra which is a deformation of the commutative algebra of functions on this manifold.
This deformation is often performed only on the level of formal power series in a small parameter, usually called Planck's constant $\hbar$.
The resulting construction goes under the name of (formal) \emph{deformation quantization}.
See e.g.~\cite{woodhouse,ali-englis,echeverria-enriquez} for an overview on geometric quantization and \cite{hawkins} for its relation to deformation quantization.

In our construction, the symplectic manifold is always $\T^*M$, the cotangent bundle to a pseudo-Riemannian manifold $M$.
We obtain a certain unique realization of deformation quantization -- a formal power series in $\hbar$ that gives an associative product on functions on $\T^*M$, which depends only on the geometry of $M$.

The plan of the paper is as follows:
In Sect.~\ref{Elements of differential geometry} we introduce notation for various objects of differential geometry of pseudo-Riemannian manifolds, which we will use when presenting our results.
In Sect.~\ref{Quantization} we define a family of quantizations depending on a $\tau\in[0,1]$.
The most important is the geodesic Weyl quantization, which corresponds to $\tau=\frac12$.
We discuss its basic properties.
Among these properties the most demanding technically is the formula for the star product.
We give its expansion up to the terms of the 4th order.
In Sect.~\ref{Methods} we explain the methods for the derivation of the expansion of the star product.
This is the most technical part of our paper.
The methods that we use are essentially known from the works of Synge~\cite{synge:ab}, DeWitt~\cite{dewitt:aa} (see also~\cite{christensen:aa}), Avramidi~\cite{avramidi:aa,avramidi:ab} and others.

There are several systems of notation in differential geometry.
We will use more than one.
For the presentation of our results, we use mostly a coordinate-free and index-free notation, which is rather concise and transparent.
It works especially well around the diagonal and is convenient for presentation of the results of our work.
However, to compute quantities it is preferable to use other notations, which typically involve coordinates and indices.

\section{Elements of differential geometry}
\label{Elements of differential geometry}

\subsection{Basic notation}

Let $M$ be a connected manifold.
The tangent and cotangent space at $x \in M$ are denoted $\T_xM$ and $\T_x^*M$, respectively.
$\T_x^{p,q}M$ will denote the space of $p$-contravariant and $q$-covariant tensors at $x$.

Often we will use a coordinate dependent notation, which involves indices, denoted by Greek letters.
Sometimes we will also use multiindices, which will be indicated by boldface letters.
For instance, $\boldsymbol\alpha = (\alpha_1, \dotsc, \alpha_n)$ and $\abs{\boldsymbol\alpha} = n$.
Thus $T\in\T_x^{p,q}M$ after fixing a system of coordinates can be written as $T_{\boldsymbol\alpha}^{\boldsymbol\beta}$ with $|{\boldsymbol\alpha}|=q$, $|{\boldsymbol\beta}|=p$.

From now on we assume that $M$ is a (pseudo-)Riemannian manifold~$M$ with the metric tensor~$g$.
For any $x \in M$ let $U_x \subset \T_xM$ be the set of vectors $u$ such that the inextendible geodesic $\gamma_u(\tau)$ starting at $x$ with initial velocity $u$ is defined at least for $\tau \in [0,1]$.
The \emph{exponential map} is then defined as
\begin{equation*}
  U_x \ni u \mapsto \exp_x(u) = \gamma_u(1) \in M.
\end{equation*}
For brevity we will often write
\begin{equation*}
  x+u \defn \exp_x(u),
  \quad
  x-u \defn \exp_x(-u).
\end{equation*}
We say that $M$ is \emph{geodesically complete} if for any $x\in M$ we have $U_x=\T_xM$.

\subsection{Bitensors}

Given two points $x,y \in M$, a \emph{bitensor} is an element $T \in \T_x^{p,q}M\otimes\T_y^{t,s}M$ for some $p,q,t,s$.
A bitensor field is a function $M \times M \ni (x,y) \mapsto T(x,y)$ such that $T(x,y)$ is a bitensor.
In other words, a bitensor field $T$ is a section of the exterior tensor product bundle $\T^{p,q}M\boxtimes\T^{t,s}M$.
Below will generally not distinguish between bitensors and bitensor fields in notation, and call both simply `bitensors'.

If we use coordinate notation, we distinguish indices belonging to the second point by primes.
For example, $T(x, y)_{\boldsymbol\alpha\,\boldsymbol\mu'}^{\boldsymbol\beta\,\boldsymbol\nu'}$ is a $|\boldsymbol\beta|$-contravariant, $|\boldsymbol\alpha|$-covariant tensor at~$x$ and $|\boldsymbol\nu|$-contravariant, $|\boldsymbol\mu|$-covariant tensor at~$y$.
Note that in the context of bitensors the prime is not a part of the name of the corresponding indices, and only the indication to which point they belong.

As another example, consider a bitensor $T(x,y)_{\mu\nu'}$ and two vector fields $v^\mu$ and $w^\mu$.
Then
\begin{equation*}
  f(x,y) = T(x,y)_{\mu\nu'} v(x)^\mu w(y)^\nu
\end{equation*}
is a biscalar, i.e., a scalar in $x$ and $y$.

If no ambiguity arises, we will often omit the dependence on $x,y$.

For the \emph{coincidence limit} $y \to x$ we use Synge's bracket notation
\begin{equation*}
  [T](x) \defn \lim_{y \to x} T(x, y),
\end{equation*}
whenever the limit exists and is independent of the path $y \to x$.

\subsection{Parallel transport and covariant derivative}
\label{Parallel transport and covariant derivative}

The metric defines the \emph{parallel transport} along an arbitrary curve.
We will most often use the parallel transport along geodesics.
Given $u \in \T_xM$ and a tensor $T \in \T_x^{p,q}M$, denote by
\begin{equation}\label{eq:parallel-bracket}
  \parap{T}^u \in \T_{x+u}^{p,q}M
\end{equation}
the tensor $T$ \emph{parallel transported} from $x$ to the point $x+u$ along the unique geodesic given by~$u$.

Let $S \in \T_{x+u}^{p,q}M$.
The tensor $S$ \emph{backward parallel transported} to $x$ is the unique $\parap{S}_u\in \T_x^{p,q}M$ such that
\begin{equation*}
  \parap[\big]{\parap{S}_u}^u=S.
\end{equation*}

In coordinates, the backward parallel transport on vectors is defined by the bitensor $g(x,x+u)^{\nu}{}_{\mu'}$ at $x$ and $x+u$, and on covectors by its inverse $g(x,x+u)_\nu{}^{\mu'}$.
More precisely, for a tensor $S_{\boldsymbol\mu}^{\boldsymbol\nu}\in\T_{x+u}^{p,q}M$ we have
\begin{equation*}
  g(x,x+u)^{\boldsymbol\alpha}{}_{\boldsymbol\mu'} g(x,x+u)_{\boldsymbol\beta}{}^{\boldsymbol\nu'} S_{\boldsymbol\nu}^{\boldsymbol\mu} = \bigl(\parap{S}_u\bigr)^{\boldsymbol\alpha}_{\boldsymbol\beta},
\end{equation*}
where
\begin{align*}
  g(x,x+u)^{\boldsymbol\alpha}{}_{\boldsymbol\mu'} &= g(x,x+u)^{\alpha_1}{}_{\mu_1'} \,\cdots\, g(x,x+u)^{\alpha_p}{}_{\mu_p'}, \quad p=|\boldsymbol\alpha|, \\
  g(x,x+u)_{\boldsymbol\beta}{}^{\boldsymbol\nu'} &= g(x,x+u)_{\beta_1}{}^{\nu_1'} \,\cdots\, g(x,x+u)_{\beta_q}{}^{\nu_q'}, \quad q=|\boldsymbol\beta|,
\end{align*}
and we use the Einstein summation convention.
We have the identities
\begin{align}
  g(x,x+u)^\alpha{}_{\mu'} g(x,x+u)_\beta{}^{\mu'} &= \delta_\beta^\alpha, \nonumber \\
  g(x)_{\alpha\beta} g(x,x+u)^\alpha{}_{\mu'} g(x,x+u)^\beta{}_{\nu'} &= g(x+u)_{\mu\nu}. \label{byby}
\end{align}
The latter identity means that the metric is covariantly constant.

Let $M \ni x \mapsto T(x) \in \T_x^{r,s}M$ be a tensor field, i.e., a tensor-valued function.
The \emph{covariant derivative} of $T$ in direction $u$ is defined as
\begin{align*}
  u \cdot \nabla T(x)
  &\defn \lim_{\tau \to 0} \frac{1}{\tau} \bigl( T(x) - \parap{T(x+\tau u)}_{\tau u} \bigr) \\
  &\mathrel{\mathrlap{\phantom\defn}=} \frac{\d}{\d\tau} \parap{T(x+\tau u)}_{\tau u} \Big|_{\tau=0}
\end{align*}
or, in coordinates,
\begin{equation*}\begin{split}
  \nabla_\mu T^{\boldsymbol\alpha}_{\boldsymbol\beta} \defn T^{\boldsymbol\alpha}_{\boldsymbol\beta;\mu} &\defn \partial_\mu T^{\boldsymbol\alpha}_{\boldsymbol\beta} + \Gamma^{\alpha_1}_{\mu\nu} T^{\nu\alpha_2\dotsm\alpha_r}_{\boldsymbol\beta} + \dotsb + \Gamma^{\alpha_r}_{\mu\nu} T^{\alpha_1\dotsm\alpha_{r-1}\nu}_{\boldsymbol\beta} \\&\quad - \Gamma^\nu_{\mu\beta_1} T^{\boldsymbol\alpha}_{\nu\beta_2\dotsb\beta_s} - \dotsb - \Gamma^\nu_{\mu\beta_s} T^{\boldsymbol\alpha}_{\beta_1\dotsb\beta_{s-1}\nu},
\end{split}\end{equation*}
where $\Gamma^\lambda_{\mu\nu}$ are the Christoffel symbols.

We can also take covariant derivatives of bitensors.
For that case, note that derivatives with respect to the two base points commute with each other.
That is (suppressing all other indices), every bitensor field $(x,y) \mapsto T(x, y)$ satisfies the identity $T_{;\mu\nu'} = T_{;\nu'\mu}$.

An important result concerning the covariant derivative of bitensors and their coincidence limit is \emph{Synge's rule}.
It states:
\begin{equation}\label{eq:synge-rule}
  [T]_{;\mu} = [T_{;\mu}] + [T_{;\mu'}].
\end{equation}
We refer to Chap.~I.4.2 of the excellent review article~\cite{poisson:ab} for a proof.

\subsection{Horizontal and vertical derivatives}

Let $\T^*M \ni (x,p) \mapsto S(x,p) \in \T_x^{r,s}M$ be a tensor-valued function on the cotangent bundle.
The \emph{horizontal derivative} of $S$ in direction $u$ is defined as
\begin{align*}
  u \cdot \nnabla S(x,p)
  &\defn \lim_{\tau \to 0} \frac{1}{\tau} \bigl( S(x,p) - \parap[\big]{S(x+\tau u, \parap{p}^{\tau u})}_{\tau u} \bigr) \\
  &\mathrel{\mathrlap{\phantom\defn}=} \frac{\d}{\d\tau} \parap[\big]{S(x+\tau u, \parap{p}^{\tau u})}_{\tau u} \Big|_{\tau=0}
\end{align*}
or, in coordinates,
\begin{equation*}\begin{split}
  \nnabla_\mu S^{\boldsymbol\alpha}_{\boldsymbol\beta} \defn S^{\boldsymbol\alpha}_{\boldsymbol\beta;\mu} &\defn \partial_\mu S^{\boldsymbol\alpha}_{\boldsymbol\beta} + \Gamma^\nu_{\mu\rho} p_\nu \frac{\partial}{\partial p_\rho} S^{\boldsymbol\alpha}_{\boldsymbol\beta} + \Gamma^{\alpha_1}_{\mu\nu} S^{\nu\alpha_2\dotsm\alpha_r}_{\boldsymbol\beta} + \dotsb + \Gamma^{\alpha_r}_{\mu\nu} S^{\alpha_1\dotsm\alpha_{r-1}\nu}_{\boldsymbol\beta} \\&\quad - \Gamma^\nu_{\mu\beta_1} S^{\boldsymbol\alpha}_{\nu\beta_2\dotsb\beta_s} - \dotsb - \Gamma^\nu_{\mu\beta_s} S^{\boldsymbol\alpha}_{\beta_1\dotsb\beta_{s-1}\nu}.
\end{split}\end{equation*}
Note that the horizontal derivative can be viewed as a natural generalization of the covariant derivative.
Therefore, it is natural to denote it by the same symbols -- it will not lead to ambiguous expressions.

The \emph{vertical derivative} of $S$ at $x$ in direction $q \in \T^*_xM$ is defined as
\begin{align*}
  q \cdot \ppartial S(x,p)
  &\defn \lim_{\tau \to 0} \frac{1}{\tau} \bigl( S(x,p) - S(x,p+\tau q) \bigr) \\
  &\mathrel{\mathrlap{\phantom\defn}=} \frac{\d}{\d\tau} S(x,p+\tau q) \Big|_{\tau=0},
\end{align*}
or, in coordinates,
\begin{equation*}
  \ppartial^\mu S^{\boldsymbol\alpha}_{\boldsymbol\beta} \defn \frac{\partial}{\partial p_\mu} S^{\boldsymbol\alpha}_{\boldsymbol\beta}.
\end{equation*}
Note that the vertical derivatives commutes with the horizontal derivative.

\subsection{Geodesically convex neighbourhood of the diagonal}

In general, a pair of points of $M$ can have no joining geodesics, a single one or many.
We will say that $M$ is \emph{geodesically simple} if every pair of points is joined by a unique geodesic $\gamma_{x,y}$.

Unfortunately, many interesting connected geodesically complete manifolds are not geodesically simple.
However, in the general case we have a weaker property, which we describe below.

Let $\Diag \defn \{(x,x) \mid x \in M\}$ denote the diagonal.
There exists a neighbourhood $\Omega \subset M \times M$ of $\Diag$ with the property: for all $(x,y) \in \Omega$ there is a unique geodesic $[0,1] \ni \tau \mapsto \gamma_{x,y}(\tau) \in M$ joining $x$ and $y$, and
\begin{equation*}
  \gamma_{x,y}([0,1]) \times \gamma_{x,y}([0,1]) \subset \Omega.
\end{equation*}
Such a neighbourhood will be called a \emph{geodesically convex neighbourhood of the diagonal}.

For $(x,y) \in \Omega$, we introduce the suggestive notation
\begin{equation}\label{eq:tangent_vector}
  (y-x) \defn \exp^{-1}_x(y) \in \T_xM,
\end{equation}
which is the tangent to the distinguished geodesic joining $x$ and $y$.
Note that $(y-x)$ is a bitensor.
More precisely, it is a vector in $x$ and a scalar in $y$, i.e., an element of $T^{1,0}M \boxtimes T^{0,0}M$.

We have
\begin{equation*}
  x + \tau(y-x) = \gamma_{x,y}(\tau).
\end{equation*}
Parallel transporting $(y-x)$, we define for $\tau \in [0,1]$
\begin{equation*}
  (y-x)_\tau \defn \parap{(y-x)}^{\tau(y-x)} \in \T_{x+\tau (y-x)}M
\end{equation*}
as a short-hand.
Clearly $(y-x) = (y-x)_0$ and
\begin{equation*}
  (y-x)_\tau = (1-\tau)^{-1} (y-z), \quad z = x+\tau (y-x)
\end{equation*}
for $\tau \neq 1$.
Furthermore, note the coincidence limit
\begin{equation}\label{eq:deriv_y-x}
  \frac{\partial(y-x)^\mu}{\partial y^\nu} \biggr|_{x=y} = \delta_\nu^\mu,
\end{equation}
which follows directly from~\eqref{eq:tangent_vector}.

\subsection{Synge's world function}

\emph{Synge's world function} $\Omega\ni (x,y) \mapsto \sigma(x, y)$ is defined as half the squared geodesic distance between~$x$ and~$y$.
That is, using the notation of \eqref{eq:tangent_vector},
\begin{equation}\label{eq:defn-synge}
  \sigma(x,y) \defn \frac12 (y-x)^\alpha (y-x)_\alpha = \frac12 (x-y)^\alpha (x-y)_\alpha.
\end{equation}
It is an example of a biscalar.

In the covariant derivatives of $\sigma$, the semicolon $;$ will usually be omitted.
Thus for any multiindices $\boldsymbol{\alpha,\beta,\mu,\nu}$
\begin{equation*}
  \sigma_{;\boldsymbol{\alpha\mu'}}^{;\boldsymbol{\beta\nu'}}
  = \sigma_{\boldsymbol{\alpha\mu'}}^{\boldsymbol{\beta\nu'}}.
\end{equation*}

We have
\begin{equation*}
  \sigma(x,y)^{\mu} = -(y-x)^\mu.
\end{equation*}

Introduce the \emph{transport operators} $\cD \defn \sigma^\mu \nabla_\mu$ and $\cD' \defn \sigma^{\mu'} \nabla_{\mu'}$.
The definition~\eqref{eq:defn-synge} of Synge's world function immediately implies
\begin{equation}\label{eq:defn-synge1}
  (\cD - 2) \sigma = 0,
  \quad
  (\cD' - 2) \sigma = 0.
\end{equation}

Differentiating~\eqref{eq:defn-synge1} once, we obtain
\begin{equation}\label{eq:synge-ident}\begin{alignedat}{2}
  (\cD - 1) \sigma^\mu &= 0,
  &\quad
  (\cD - 1) \sigma^{\mu\rlapprime} &= 0,
  \\
  (\cD' - 1) \sigma^\mu &= 0,
  &\quad
  (\cD' - 1) \sigma^{\mu\rlapprime} &= 0,
\end{alignedat}\end{equation}
which are useful identities for the calculation of the coincidence limits of derivatives of the world function, see Subsect.~\ref{sec:synge-coinc}.

\subsection{Bitensor of parallel transport}

Another example of a bitensor is the \emph{bitensor of the parallel transport}, $g^\mu{}_{\nu'}$, which transports
$\T_yM$ onto $\T_xM$ along the geodesics $\gamma_{y,x}$.
In the notation of~\eqref{eq:parallel-bracket}, for a vector field~$v$,
\begin{equation*}
  g(x,y)^\mu{}_{\nu'} v(y)^\nu = \bigl(\parap{v}_{(y-x)}\bigr)^\mu(x).
\end{equation*}
Note that in Subsect.~\ref{Parallel transport and covariant derivative} we considered a similar object $g_\mu{}^\nu(x,x+u)$, except that there it was viewed as a function of $x\in M$, $u\in\T_xM$, and now it is viewed as a function of $(x,y) \in \Omega$.

Equivalently, the bitensor $g^\mu{}_{\nu'}$ is defined by the transport equations
\begin{equation}\label{eq:covd_parallel_transport}
  \cD g^\mu{}_{\nu'} = 0 = \cD' g^\mu{}_{\nu'}
\end{equation}
with the initial condition $[g^\mu{}_{\nu'}] = \delta^\mu{}_\nu$.

\subsection{Van Vleck--Morette determinant}

If $T\in \T_x^{1,1}M$, we will write $|T|$ for $\abs{\det T}$.
Note that $|T|$ is well defined independently of coordinates.

If $S \in \T_x^{0,2}M$ or $S \in \T_x^{2,0}M$, we will use the same notation $|S|$ for the absolute value of the determinant.
$|S|$ may now depend on the coordinates, but it can still be a useful object.
For instance, the metric $g(x)$ belongs to $\T_x^{0,2}$, however $|g(x)|$ will play a considerable role in our analysis.

Let $T(x,y)^\mu{}_{\nu'}$ be a bitensor, contravariant in~$x$ and covariant in~$y$.
Then it is easy to see that
\begin{equation*}
  |T(x,y)| \frac{|g(x)|^\frac12}{|g(y)|^\frac12}
\end{equation*}
does not depend on coordinates, and hence is a biscalar.
For instance, by~\eqref{byby},
\begin{equation}\label{coord}
  \bigl|g(x,y)^\mu{}_{\nu'}\bigr| \frac{\abs{g(x)}^\frac12}{\abs{g(y)}^{\frac12}}=1.
\end{equation}

For $(x,y) \in \Omega$, the matrix
\begin{equation}\label{bite}
  \frac{\partial(y-x)^\mu}{\partial y^{\nu'}} = \sigma(x,y)^{\mu}{}_{\nu'}
\end{equation}
is a bitensor, contravariant in~$x$ and covariant in~$y$.
Note that the derivatives on the right-hand side of~\eqref{bite} can be applied in any order, and instead of the covariant derivatives we can use the usual derivatives.

The \emph{Van Vleck--Morette determinant} is defined as
\begin{equation*}
  \Delta(x,y) \defn \abs[\bigg]{\frac{\partial(y-x)}{\partial y}} \frac{|g(x)|^\frac12}{|g(y)|^\frac12}.
\end{equation*}
By the discussion above, its definition does not depend on the choice of coordinates.
For coinciding points, we have $\Delta(x,x)=1$ by~\eqref{eq:deriv_y-x}.

\begin{theorem}
  $\Delta$ is continuous on $\Omega$ and $\Delta(x,y)=\Delta(y,x)$.

  Besides, for any $\tau \in [0,1]$, we have
  \begin{subequations}\label{eq:vanvleck-morette-ident}\begin{align}
    \Delta(x,y)
    &= \abs*{\frac{\partial(y-x)_\tau}{\partial y}} \frac{\abs[\big]{g\bigl(x + \tau(y-x) \bigr)}^\frac12}{\abs{g(y)}^\frac12}  \label{eq:vanvleck-morette-ident1} \\
    &= \abs*{\frac{\partial(x-y)_\tau}{\partial x}} \frac{\abs[\big]{g\bigl(y + \tau(x-y) \bigr)}^\frac12}{\abs{g(x)}^\frac12}. \label{eq:vanvleck-morette-ident2}
  \end{align}\end{subequations}
\end{theorem}
\begin{proof}
  We use the symmetry of the world-function and then we raise and lower the indices:
  \begin{align*}
    \sigma(x,y)^\mu{}_{\nu'}
    &= \sigma(y,x)_\nu{}^{\mu'} \\
    &= \sigma(y,x)^\alpha{}_{\beta'} g(x)_{\alpha\nu} g(y)^{\beta'\mu'}.
  \end{align*}
  Therefore,
  \begin{align*}
    \bigl|\sigma(x,y)^\mu{}_{\nu'}\bigr|
    &= \bigl|\sigma(y,x)^\alpha{}_{\beta'}\bigr| \frac{|g(x)|}{|g(y)|},
  \end{align*}
  which shows the symmetry of $\Delta(x,y)$.

  We have $(y-x)_\tau=\parap{(y-x)}^{\tau(y-x)}$.
  Therefore,
  \begin{equation*}
    \frac{\partial (y-x)_\tau}{\partial y} = \parap*{\frac{\partial (y-x)}{\partial y}}^{\tau(y-x)},
  \end{equation*}
  where we only parallel transport the part of the bitensor at~$x$, that is,
  \begin{equation*}
    \frac{\partial (y-x)_\tau{}^\mu}{\partial y^{\nu'}} = g(x+\tau(y-x),x)^\mu{}_\alpha \frac{\partial (y-x)^\alpha}{\partial y^{\nu'}}
  \end{equation*}
  Using~\eqref{coord}, we obtain for the determinant
  \begin{equation*}
    \abs*{\frac{\partial (y-x)_\tau}{\partial y}} = \abs*{\frac{\partial (y-x)}{\partial y}} \frac{\abs[\big]{g\bigl(x+\tau(y-x)\bigr)}^\frac12}{\abs{g(x)}^\frac12}.
  \end{equation*}
  Now~\eqref{eq:vanvleck-morette-ident1} and \eqref{eq:vanvleck-morette-ident2} follow.
\end{proof}

The Van Vleck--Morette determinant enters in the expression for the Jacobian of a certain useful change of coordinates:
\begin{proposition}
  Set $z_\tau = x + \tau(y-x)$, and $u_\tau = (y-x)_\tau$.
  Then
  \begin{equation}\label{eq:measure-trafo3}
    \biggl|\frac{\partial(z_\tau,u_\tau)}{\partial(x,y)}\biggr| = \Delta(x,y) \frac{|g(x)|^\frac12 |g(y)|^\frac12}{|g(z_\tau)|}.
  \end{equation}
\end{proposition}
\begin{proof}
  We have
  \begin{equation}\label{inse}
    \biggl|\frac{\partial(z_\tau, u_\tau)}{\partial(x,y)}\biggr| = \biggl|\frac{\partial(z_\tau, u_\tau)}{\partial(x,u_\tau)}\biggr|\,\biggl|\frac{\partial(x,u_\tau)}{\partial(x,y)}\biggr| = \biggl|\frac{\partial z_\tau}{\partial x}\Big|_{u_\tau}\biggr|\,\biggl|\frac{\partial u_\tau}{\partial y}\Big|_x\biggr|.
  \end{equation}
  Now, by \eqref{eq:vanvleck-morette-ident},
  \begin{equation}\label{eq:measure-trafo1}
    \biggl|\frac{\partial u_\tau}{\partial y}\Big|_x\biggr| = \Delta(x,y) \frac{|g(y)|^\frac12}{|g(z_\tau)|^\frac12}
  \end{equation}
  Clearly,
  \begin{equation}\label{eq:measure-trafo2}
    \biggl|\frac{\partial z_\tau}{\partial x}\biggr| = \frac{|g(x)|^\frac12}{|g(z_\tau)|^\frac12}.
  \end{equation}
  Inserting~\eqref{eq:measure-trafo1} and~\eqref{eq:measure-trafo2} into~\eqref{inse} we obtain~\eqref{eq:measure-trafo3}.
\end{proof}

The Jacobian~\eqref{eq:measure-trafo3} will play an important role in our construction.
For brevity we therefore define the geometric factor ($z_\tau = x + \tau(x-y)$ as before)
\begin{equation*}
  \Upsilon_\tau(x,y) \defn \frac{\Delta(x,y)^\frac12 \abs{g(x)}^\frac14 \abs{g(y)}^\frac14}{\abs{g(z_\tau)}^\frac12},
\end{equation*}
which will appear in several of our proofs.

\begin{remark}
  Actually, the left hand side of \eqref{eq:measure-trafo3} depends only on the connection on $M$, and does not depend on its pseudo-Riemannian structure.
  Clearly, the same is true concerning $\Upsilon_\tau$.
  Therefore, the balanced geodesic Weyl quantization can be defined on a manifold with just a connection -- it does not need to be a pseudo-Riemannian manifold.
  (This remark is due to Wojciech Kami\'nski.)
\end{remark}

\section{Quantization}
\label{Quantization}

\subsection{Quantization on a flat space}

In this subsection we collect well-known facts concerning the quantization on a flat space, which we would like to generalize to curved (pseudo-)Riemannian manifolds.

Consider the vector space $\cX=\rr^d$, with $\cX^\#$ denoting the space dual to~$\cX$.
By the Schwartz Kernel Theorem, continuous operators $B : \cS(\cX) \to \cS'(\cX)$ are defined by their kernels $B(\cdot\,,\cdot) \in \cS'(\cX \times \cX)$.
An operator $B$ on $L^2(\cX)$ is Hilbert--Schmidt if and only if its kernel satisfies $B(\cdot\,,\cdot)\in L^2(\cX \times \cX)$.
The Hilbert--Schmidt scalar product satisfies the identity
\begin{equation*}
  \Tr(A^*B) = \int_{\cX \times \cX} \bar{A(x,y)} B(x,y) \d x \d y.
\end{equation*}

Let us fix a positive number $\hbar$ called \emph{Planck's constant}.
The parameter $\hbar$ conveniently keeps track of the ``order of semiclassical approximation''.

Let $b \in \cS'(\cX \times \cX^\#)$.
For any $\tau \in \rr$, we associate with $b$ the operator $\Op_\tau(b) : \cS(\cX) \to \cS'(\cX)$, given by the kernel
\begin{equation*}
  \Op_\tau(b)(x,y) \defn \int_{\cX^\#} b\bigl(x+\tau(y-x),p\bigr) \e^{\frac{\i (x-y) \cdot p}{\hbar}} \frac{\d p}{(2\uppi\hbar)^d}
\end{equation*}
and called the \emph{$\tau$-quantization of the symbol~$b$}.

The most natural quantization corresponds to $\tau=\frac12$ and is called the \emph{Weyl quantization of the symbol~$b$}.
Instead of $\Op_\frac12(b)$ we will simply write $\Op(b)$.

The quantizations corresponding to $\tau=0$ and $\tau=1$ are also useful.
They are sometimes called the $x,p$ and the $p,x$ \emph{quantizations}.
In a part of the PDE literature, the $x,p$ quantization is treated as the standard one.
Quantizations corresponding to $\tau$ different from $0,\frac12,1$ are of purely academic interest.

If $A,B$ are Hilbert--Schmidt operators on $L^2(M)$ such that $A=\Op_\tau(a_\tau)$ and $B=\Op_\tau(b_\tau)$, then
\begin{equation*}
  \Tr(A^*B) = \int_{\cX \times \cX^\#} \bar{a_\tau(z,p)} b_\tau(z,p) \d z\, \frac{\d p}{(2\uppi\hbar)^d}.
\end{equation*}

Suppose that the operators $\Op(a)$ and $\Op(b)$ can be composed.
Then one defines the \emph{star product} or the \emph{Moyal product} $a\star b$ by
\begin{equation*}
  \Op(a \star b) = \Op(a) \Op(b).
\end{equation*}

Using the identity
\begin{equation}\begin{split}\label{eq:gaussian}\MoveEqLeft
  \exp\bigl(\tfrac12 \partial_x \cdot A\, \partial_x\bigr) f(x) \\&= (2\uppi)^{-\frac{n}{2}} (\det A)^{-\frac12} \int_{\rr^n} \exp\bigl(-\tfrac12 (x-y) \cdot A^{-1} (x-y)\bigr) f(y) \d y
\end{split}\end{equation}
for an invertible $n \times n$ matrix $A$ with $\Re A \geq 0$, which holds for $f \in \mathcal{S}(\rr^n)$ but can also be understood in larger generality, one obtains two formulas for the star product:
\begin{subequations}\label{eq:moyal}\begin{align}
  \begin{split}
    (a \star b)(z,p)
    &= \int_{\cX \times \cX \times \cX^\# \times \cX^\#} \e^{\frac{2\i (u_1 \cdot p_2 - u_2 \cdot p_1)}{\hbar}} a(z+u_1,p+p_1) \\&\quad\times b(z+u_2,p+p_2) \d u_1 \d u_2 \frac{\d p_1 \d p_2}{(\uppi\hbar)^d}
  \end{split} \\
  \begin{split}
    &= \exp\bigl(\tfrac{\i}{2}\hbar(\partial_{u_1} \cdot \partial_{p_2} - \partial_{u_2} \cdot \partial_{p_1})\bigr) \\&\quad\times a(z+u_1,p+p_1) b(z+u_2,p+p_2) \Big|_{\substack{u_1=u_2=0\\p_1=p_2=0}}.
  \end{split}
\end{align}\end{subequations}

\subsection{Operators on a manifold}

In this subsection, the (pseudo-)Riemannian structure of $M$ is irrelevant.

If $B$ is a continuous operator $C_\mathrm{c}^\infty(M) \to \cD'(M)$, then its \emph{kernel} is the distribution in $\cD'(M\times M)$, denoted $B(\cdot\,,\cdot)$, such that
\begin{equation*}
  \langle f | Bg \rangle = \int_{M \times M} f(x) B(x,y) g(y) \d x \d y,
  \quad
  f, g \in C_\mathrm{c}^\infty(M).
\end{equation*}
For instance, the kernel of the identity is given by the delta distribution.

We will treat elements of $C_\mathrm{c}^\infty(M)$ not as scalar functions, but as half-densities.
With this convention, the kernel of an operator is a half-density on $M \times M$.
Note that with our conventions we need not specify a density with respect to which we integrate.

If two operators $A,B$ can be composed, then we have
\begin{equation}
  AB(x,y) = \int_M A(x,z) B(z,y) \d z.\label{twocom}
\end{equation}
\eqref{twocom} is true, e.g., if both $A$ and $B$ are Hilbert-Schmidt, but of course it also holds in various other situations.

Clearly, the space of square-integrable half-densities on $M$ forms a Hilbert space which will be denoted $L^2(M)$.
It is well known that an operator $B$ on $L^2(M)$ is Hilbert--Schmidt if and only if its kernel satisfies
\begin{equation*}
  B(\cdot\,,\cdot) \in L^2(M \times M).
\end{equation*}
If two operators $A,B$ are Hilbert--Schmidt, the Hilbert--Schmidt scalar product is given by
\begin{equation}\label{eq:trace-kernel}
  \Tr(A^* B) = \int_{M \times M} \bar{A(x,y)} B(x,y) \d x \d y.
\end{equation}

\subsection{Balanced geodesic quantization}

Consider a smooth function
\begin{equation}\label{eq:funk}
  \T^*M \ni (z,p) \mapsto b(z,p).
\end{equation}
Note that $\T^*M$ possesses a natural density, independent of the (pseudo{-\penalty0\hskip0pt\relax})Riemannian structure.
Hence $b$ can be interpreted as we like -- as a scalar, density or, which is the most relevant interpretation for us, as a half-density.

Assume first that $M$ is geodesically simple.
Let $\tau\in\rr$.
We associate with~\eqref{eq:funk} an operator with the kernel
\begin{equation}\label{eq:symbol-to-kernel1}
  \Op_\tau(b)(x,y) \defn  \frac{\Delta(x,y)^\frac12 \abs{g(x)}^\frac14 \abs{g(y)}^\frac14} {\abs{g(z_\tau)}^{\frac12}} \int_{\T_{z_\tau}^*M} b(z_\tau,p) \e^{-\frac{\i u_\tau \cdot p}{\hbar}} \frac{\d p}{(2\uppi\hbar)^d},
\end{equation}
where $z_\tau = x+\tau(y-x) \in M$ and $u_\tau = (y-x)_{\tau} \in\T_{z_\tau}M$.
We call $\Op_\tau(b)$ the \emph{geodesic $\tau$-quantization of the symbol $b$}.

If $M$ is not geodesically simple, then the definition~\eqref{eq:symbol-to-kernel1} needs to be modified.
Recall that $\Omega$ is a geodesically convex neighbourhood of $\Diag$.
Choose $\Omega_1$, another geodesically convex neighbourhood of $\Diag$ such that the closure of $\Omega_1$ is contained in $\Omega$.
Fix a function $\chi\in C^\infty(M\times M)$ such that $\chi=1$ on $\Omega_1$ and $\supp \chi\subset \Omega$.
Then instead of~\eqref{eq:symbol-to-kernel1} we set
\begin{equation}\label{eq:symbol-to-kernel}
  \Op_\tau(b)(x,y) \defn \chi(x,y) \frac{\Delta(x,y)^\frac12 \abs{g(x)}^\frac14 \abs{g(y)}^\frac14} {\abs{g(z_\tau)}^{\frac12}} \int_{\T_{z_\tau}^*M} b(z_\tau,p) \e^{-\frac{\i u_\tau \cdot p}{\hbar}} \frac{\d p}{(2\uppi\hbar)^d},
\end{equation}
where $z_\tau = x+\tau(y-x) \in M$ and $u_\tau = (y-x)_{\tau} \in\T_{z_\tau}M$.

Most of the time we will use $\tau=\frac12$, which is the analog of the Weyl quantization, and then we will write simply $\Op(b)$ instead of $\Op_\frac12(b)$.
$\Op(b)$ will be called the \emph{balanced geodesic Weyl quantization of the symbol $b$}.

Our quantization depends on a \emph{Planck constant} $\hbar>0$.
This is however a minor thing: if we set $\hbar=1$, we can easily put it back in all formulas, by dividing all momenta except those appearing in the arguments of symbols by $\hbar$ (i.e., replacing $p$ by $\hbar^{-1} p$).
Therefore, for simplicity, in all proofs we will set $\hbar=1$.
However, in the statements of various properties, we will keep $\hbar$ explicit.

Conversely, suppose that we are given a kernel $B(x,y)$ supported in $\Omega_1$.
Note that then we can drop $\chi$ from~\eqref{eq:symbol-to-kernel} and its $\tau$-symbol is
\begin{equation}\label{eq:kernel-to-symbol}
  b_\tau(z,p) \defn \int_{\T_zM} \frac{|g(z)|^{\frac12}}{\Delta(x_\tau,y_\tau)^{\frac12}|g(x_\tau)|^{\frac14}|g(y_\tau)|^{\frac14}} B(x_\tau,y_\tau) \e^{\frac{\i u \cdot p}{\hbar}} \d u,
\end{equation}
where $x_\tau = z-\tau u$ and $y_\tau = z+(1-\tau)u$.

Using~\eqref{eq:measure-trafo1}, this can be reexpressed as
\begin{align*}
  b_\tau(z,p) &= \int_M \Delta(x_\tau,y)^{\frac12} \frac{\abs{g(y)}^\frac14}{\abs{g(x_\tau)}^\frac14} B(x_\tau,y) \e^{\frac{\i u_\tau \cdot p}{\hbar}} \d y
\intertext{for $\tau\neq1$, where $u_\tau=(1-\tau)^{-1}(z-y)$ and $x_\tau=z-\tau u$, or}
  b_\tau(z,p) &= \int_M \Delta(x,y_\tau)^{\frac12} \frac{\abs{g(x)}^\frac14}{\abs{g(y_\tau)}^\frac14} B(x,y_\tau) \e^{\frac{\i u_\tau \cdot p}{\hbar}} \d x
\end{align*}
for $\tau\neq0$, where $u_\tau=-\tau^{-1}(z-x)$ and $y_\tau=z+(1-\tau)u$.

\begin{proposition}
  Eq.~\eqref{eq:kernel-to-symbol} is the inverse to \eqref{eq:symbol-to-kernel}.
\end{proposition}
\begin{proof}
  On the one hand, if $b(z,p)$ is a symbol with $\tau$-quantization $\Op_\tau(b)(x,y)$,
  \begin{align*}\MoveEqLeft
    \int_{\T_zM} \Upsilon_\tau(x,y)^{-1} \Op_\tau(b)(x,y) \e^{\i u \cdot p} \d u \\
    &= \int_{\T_zM \times \T^*_zM} b(z,q) \e^{\i u \cdot (p-q)} \frac{\d q}{(2\uppi)^d}\, \d u
    = b(z,p),
  \end{align*}
  where $x = z-\tau u$ and $y = z+(1-\tau)u$.
  On the other hand, if $B(x,y)$ is a kernel and $b_\tau(z,p)$ is its symbol given by~\eqref{eq:kernel-to-symbol},
  \begin{align*}\MoveEqLeft
    \Upsilon_\tau(x,y) \int_{\T_z^*M} b_\tau(z,p) \e^{-\i u \cdot p} \frac{\d p}{(2\uppi)^d} \\
    &= \Upsilon_\tau(x,y) \int_{\T_z^*M \times \T_zM} \Upsilon_\tau(x',y')^{-1} B(x',y') \e^{\i (v-u) \cdot p} \d v \,\frac{\d p}{(2\uppi)^d} \\
    &= B(x,y),
  \end{align*}
  where $u = (y-x)_\tau$, $x' = z-\tau v$ and $y' = z+(1-\tau)v$.
\end{proof}

\begin{remark}
  The drawback of~\eqref{eq:symbol-to-kernel} is the fact that $\Op_\tau(b)$ depends on the cutoff $\chi$.
  It is possible to modify this definition so that it is purely geometric and this cutoff is not needed.
  In fact, recall that given $z \in M$ and $u \in \T_zM$ there exist $-\infty\leq t_-<t_+\leq+\infty$ such that $\mathopen{]}t_-,t_+\mathclose{[}\ni t\mapsto \exp_z(tu)\in M$ is an inextendible geodesics.
  We set
  \begin{equation*}
    \Op_{\tau,\mathrm{global}}(b)(x,y) \defn
    \sum_{(z,u) \in \Gamma_\tau(x,y)}
    \biggl|\frac{\partial(x,y)}{\partial(z,u)}\biggr|^{-\frac12} \int_{\T_{z}^*M} b(z,p) \e^{-\frac{\i u \cdot p}{\hbar}} \frac{\d p}{(2\uppi\hbar)^d},
  \end{equation*}
  where
  \begin{equation*}
    \Gamma_\tau(x,y) \defn \bigl\{ (z,u) \in \T M \;\big|\;
    \exp_z(-\tau u)=x,
    \exp_z\bigl((1-\tau)u\bigr)=y \bigr\}.
  \end{equation*}
  If $M$ is geodesically simple, so that we can take $\chi=1$, then
  \begin{equation*}
    \Op_{\tau}(b) = \Op_{\tau,\mathrm{global}}(b).
  \end{equation*}
  If not, they may be different, and besides $\Op_{\tau,\mathrm{global}}(b)$ does not depend on the cutoff~$\chi$.
  However, since we are mostly interested in properties of operators around the diagonal, we keep the definition~\eqref{eq:symbol-to-kernel}.
\end{remark}

\subsection{Independence of the quantization on the cutoff}

We have already mentioned that the dependence of the quantization on the cutoff $\chi$ is mild. This is of course not true if the symbol has low regularity. It can be however expected  for various typical classes of symbols used in the pseudodifferential calculus. In this subsection we will illustrate this independence on the example of the most popular symbol class, $S^m(\T^*M)$.

Let $m\in\mathbb{R}$.
The class $S^m(\T^*M)$ consists of $b\in C^\infty(T^*M)$ such that for any compact $K \subset M$ and for arbitrary multiindices $\boldsymbol\alpha, \boldsymbol\beta$
\begin{equation}\label{symbol}
  \sup_{(z,p) \in \T^*M, z\in K} \langle p \rangle^{\abs{\boldsymbol\alpha} - m} \abs{\ppartial^{\boldsymbol\alpha} \nabla_{\boldsymbol\beta} b(z,p)} < \infty.
\end{equation}

Note that the class $S^m(T^*M)$ is purely geometric (it does not depend on the choice of coordinates).


The following two propositions are versions of well known properties of the standard pseudodifferential calculus adapted to the balanced geodesic Weyl quantization.

\begin{proposition}
  Let  $b\in S^m(\T^*M)$ for some $m$.
  Then $\Op_\tau(b)(x,y)$ is smooth outside of the diagonal.
  Besides, together with all derivatives, it is $O(\hbar^\infty)$ outside of the diagonal.
  \label{diago}
\end{proposition}
\begin{proof}
  The function $(x,y) \mapsto \bigl(z_\tau(x,y),u_\tau(x,y)\bigr)$ is smooth.
  So are the cutoff function and the geometric prefactor.
  Since we are only interested in smoothness outside of the diagonal, i.e., for $u \neq 0$, it is therefore enough to study the smoothness of the function
  \begin{equation*}
    (z,u) \mapsto u^{\boldsymbol\alpha} \int b(z,p) \e^{\frac{\i u \cdot p}{\hbar}} \frac{\d p}{(2\uppi\hbar)^d}
  \end{equation*}
  for some (sufficiently large) $\abs{\boldsymbol\alpha}$.
  Now,
  \begin{align}
    \MoveEqLeft \nabla_{\boldsymbol\beta} {(-\i\partial_u)}_{\boldsymbol\gamma} u^{\boldsymbol\alpha} \int b(z,p) \e^{\frac{\i u \cdot p}{\hbar}} \frac{\d p}{(2\uppi\hbar)^d} \notag \\
    &= \hbar^{\abs{\boldsymbol\alpha}-\abs{\boldsymbol\gamma}} \int \nabla_{\boldsymbol\beta} b(z,p) (-\i \ppartial)^{\boldsymbol\alpha} p_{\boldsymbol\gamma} \e^{\frac{\i u \cdot p}{\hbar}} \frac{\d p}{(2\uppi\hbar)^d} \notag \\
    &= \hbar^{\abs{\boldsymbol\alpha}-\abs{\boldsymbol\gamma}} \int \e^{\frac{\i u \cdot p}{\hbar}} p_{\boldsymbol\gamma} (\i \ppartial)^{\boldsymbol\alpha} \nabla_{\boldsymbol\beta} b(z,p) \frac{\d p}{(2\uppi\hbar)^d}. \label{laba}
  \end{align}
  Integrand (and prefactor) of \eqref{laba} grow at most as $\hbar^{\abs{\boldsymbol\alpha}-\abs{\boldsymbol\gamma}-d} \jnorm{p}^{m-\abs{\boldsymbol\alpha}+\abs{\boldsymbol\gamma}}$.
  For large enough $\abs{\boldsymbol\alpha}$ it is thus integrable.
  This shows the smoothness of $\Op_\tau(b)(x,y)$ outside of the diagonal and that it is $O(\hbar^\infty)$.
\end{proof}


\begin{proposition}
  Suppose that $\chi_1$, $\chi_2$ are two cutoffs of the type described above and $\Op_1$, $\Op_2$ be the quantizations corresponding to $\chi_1$, $\chi_2$.
  Let $b\in S^m(\T^*M)$.
  Then $\Op_1(b)(x,y)-\Op_2(b)(x,y)$ is smooth, and together with all its derivatives it is $O(\hbar^\infty)$.
\end{proposition}
\begin{proof}
  We repeat the arguments of Prop.~\ref{diago}.
\end{proof}

\subsection{Hilbert--Schmidt scalar product}

\begin{proposition}
  Suppose that $A,B$ are Hilbert--Schmidt operators on $L^2(M)$, whose kernels are supported in~$\Omega_1$.
  Let $A=\Op_\tau(a_\tau)$ and $B=\Op_\tau(b_\tau)$.
  Then
  \begin{equation}\label{trace}
    \Tr(A^* B) = \int_{\T^*M} \conj{a_\tau(z,p)} b_\tau(z,p) \frac{\d p}{(2\uppi\hbar)^d} \d z.
  \end{equation}
\end{proposition}
\begin{proof}
  Inserting~\eqref{eq:symbol-to-kernel} in~\eqref{eq:trace-kernel}, we find
  \begin{align*}
    \Tr(A^* B) &= \int_{\T^*_{z_\tau}M \times \T^*_{z_\tau}M \times M \times M} \Upsilon_\tau(x,y)^2 \conj{a_\tau(z_\tau,p)} b_\tau(z_\tau,q) \\&\quad\times \e^{\i u_\tau \cdot (p-q)} \frac{\d p}{(2\uppi)^d} \frac{\d q}{(2\uppi)^d} \d x \d y.
  \end{align*}
  Changing integration variables from $(x,y)$ to $(z_\tau,u_\tau)$, see~\eqref{eq:measure-trafo3}, and dropping the subscript $\tau$ gives
  \begin{align*}
    \text{R.H.S.}
    &= \int_{\T^*_zM \times \T^*_zM \times \T_zM \times M} \conj{a_\tau(z,p)} b_\tau(z,q) \e^{\i u \cdot (p-q)} \frac{\d p}{(2\uppi)^d} \frac{\d q}{(2\uppi)^d} \d u \d z \\
    &= \int_{\T^*M} \conj{a_\tau(z,p)} b_\tau(z,p) \frac{\d p}{(2\uppi)^d} \d z.
    \qedhere
  \end{align*}
\end{proof}

Thus all quantizations that we defined are unitary (up to a natural coefficient), which, as we believe, is a strong argument in favor of them.

One can expect that the formula~\eqref{trace} holds for a large class of operators such that $A^*B$ is trace class, even if $A$ and $B$ are not Hilbert--Schmidt.
For example, denote by $f(x)$ the operator of multiplication by a complex function $M \ni z \mapsto f(z)$.
Note that $\Op_\tau(f) = f(x)$.
Therefore, we obtain the formula
\begin{equation}\label{trace-f}
  \Tr\bigl(f(x) \Op_\tau(b)\bigr) = \int_{\T^*M} f(z) b(z,p) \frac{\d p}{(2\uppi\hbar)^d} \d z.
\end{equation}

Note that the integral kernel of the operator $f(x)$ is supported exactly at the diagonal, therefore in the identity~\eqref{trace-f} there is no dependence on the cutoff $\chi$.

\subsection{Translating between different quantizations}

Suppose we are given a symbol for the geodesic $\tau$-quantization and wish to find a corresponding symbol for the $\tau'$-quantization.
An asymptotic formula is given by the following proposition (see also~\cite{safarov}, where an analogous formula is derived):

\begin{proposition}
  Let $\tau, \tau' \in \rr$ and consider symbols $b_\tau,b_{\tau'}$ such that
  \begin{equation*}
    \Op_\tau(b_\tau) \sim \Op_{\tau'}(b_{\tau'}).
  \end{equation*}
  Then
  \begin{align*}
    b_\tau(z,p)
    &\sim \exp\bigl(-\i\hbar(\tau'-\tau)\ppartial\cdot\nnabla\bigr) b_{\tau'}(z,p) \\
    &= \sum_n \frac{(\tau'-\tau)^n}{n!} (-\i\hbar\ppartial\cdot\nnabla)^n b_{\tau'}(z, p).
  \end{align*}
\end{proposition}
\begin{proof}
  Let $u \in \T_zM$, $p \in \T_z^*M$ and $q' \in \T_{z'}M$ with
  \begin{equation*}
    z' \defn z+(\tau'-\tau)u.
  \end{equation*}
  Then, setting
  \begin{align*}
    u' &= \parap{u}^{(\tau'-\tau)u} \in \T_{z'}M, \\
    q  &= \parap{q'}_{(\tau'-\tau)u} \in \T^*_zM,
  \end{align*}
  we calculate
  \begin{align*}
    b_\tau(z,p)
    &= \int_{\T_zM} \frac{\abs{g(z)}^\frac12}{\abs{g(z')}^\frac12} \int_{\T^*_{z'}M} b_{\tau'}(z',q') \e^{\i (u \cdot p - u' \cdot q')} \frac{\d q'}{(2\uppi)^d} \d u \\
    &= \int_{\T_zM \times \T^*_zM} b_{\tau'}(z',\parap{q}^{(\tau'-\tau)u}) \e^{\i u \cdot (p-q)} \frac{\d q}{(2\uppi)^d} \d u \\
    &\sim \int_{\T_zM \times \T^*_zM} \e^{\i u \cdot (p-q)} \sum_{\boldsymbol\alpha} \frac{1}{\abs{\boldsymbol\alpha}!} (\tau'-\tau)^{\abs{\boldsymbol\alpha}} u^{\boldsymbol\alpha} \nnabla_{\boldsymbol\alpha} b_{\tau'}(z,q) \frac{\d q}{(2\uppi)^d} \d u \\
    &= \int_{\T_zM \times \T^*_zM} \sum_{\boldsymbol\alpha} \frac{1}{\abs{\boldsymbol\alpha}!} \bigl( (\i \partial_q{})^{\boldsymbol\alpha} \e^{\i u \cdot (p-q)} \bigr) (\tau'-\tau)^{\abs{\boldsymbol\alpha}} \nnabla_{\boldsymbol\alpha} b_{\tau'}(z,q) \frac{\d q}{(2\uppi)^d} \d u \\
    &= \int_{\T_zM \times \T^*_zM} \e^{\i u \cdot (p-q)} \sum_{\boldsymbol\alpha} \frac{1}{\abs{\boldsymbol\alpha}!} \bigl( -\i (\tau'-\tau) \partial_q \cdot \nnabla \bigr)^{\abs{\boldsymbol\alpha}} b_{\tau'}(z,q) \frac{\d q}{(2\uppi)^d} \d u \\
    &= \exp\bigl(-\i(\tau'-\tau)\ppartial\cdot\nnabla\bigr) b_{\tau'}(z,p).
    \qedhere
  \end{align*}
\end{proof}

\subsection{Quantization of polynomial symbols}

In this section, we consider the quantization of symbols, which are polynomial in the momenta.
Note that for polynomial symbols the integral kernel of their quantization is supported on the diagonal, therefore there are no problems with multiple geodesics between two points.

As usual, boldface Greek letters denote multiindices.
We sum over repeated multiindices and write $p_{\boldsymbol\alpha} = p_{\alpha_1} \dotsm p_{\alpha_n}$.

\begin{proposition}\label{prop:monomial_quant}
  Consider the polynomial symbol
  \begin{equation}\label{eq:polypo}
    a(z,p) = a(z)^{\boldsymbol\alpha} p_{\boldsymbol\alpha},
  \end{equation}
  where $a(z)^{\boldsymbol\alpha}$ are symmetric tensors.
  Then
  \begin{equation}\label{eq:monomial_quant}\begin{split}
    \Op_\tau(a) &= \hbar^{|\boldsymbol\gamma|+|\boldsymbol\delta|} \tau^{|\boldsymbol\alpha+\boldsymbol\gamma|}(1-\tau)^{|\boldsymbol\beta+\boldsymbol\delta|} \\&\quad\times \abs{g}^{-\frac14} (-\i\hbar\nabla)_{\boldsymbol\alpha} a^{\boldsymbol\alpha+\boldsymbol\beta+\boldsymbol\gamma+\boldsymbol\delta} \xi_{\boldsymbol\gamma\boldsymbol\delta} \abs{g}^\frac12 (-\i\hbar\nabla)_{\boldsymbol\beta} \abs{g}^{-\frac14},
  \end{split}\end{equation}
  with
  \begin{equation*}
    \xi_{\boldsymbol\gamma\boldsymbol\delta} \defn (-1)^{\abs{\boldsymbol\gamma}} {(-\i\nabla)}_{\boldsymbol\gamma} {(-\i\nabla')}_{\boldsymbol\delta} \Delta(x,y)^{-\frac12} \Big|_{x=y},
  \end{equation*}
  where the prime indicates (covariant) differentiation with respect to $y$.
\end{proposition}
\begin{proof}
  We calculate
  \begin{align*}
    \begin{split}
      \bigl(f \,\big| \Op_\tau(a) h\bigr)
      &= \int_{M\times M \times \T^*_{x+\tau(y-x)}M} \bar{f(x)} a\bigl(x+\tau(y-x)\bigr)^{\boldsymbol\alpha} p_{\boldsymbol\alpha} h(y) \\&\quad\times \Upsilon_\tau(x,y) \e^{-\i u\cdot p} \d x \d y \,\frac{\d p}{(2\uppi)^d}
    \end{split} \\
    \begin{split}
      &= \int_{M\times \T_zM \times \T^*_zM} \bar{f(z-\tau u)} a(z)^{\boldsymbol\alpha} p_{\boldsymbol\alpha} h\bigl(z+(1-\tau)u\bigr) \\&\quad\times \Upsilon_\tau\bigl(z-\tau,z+(1-\tau)u\bigr)^{-1} \e^{-\i u\cdot p} \d z \d u \,\frac{\d p}{(2\uppi)^d}
    \end{split} \\
    \begin{split}
      &= \int_{M\times \T_zM \times \T^*_zM} \bar{f(z-\tau u)} a(z)^{\boldsymbol\alpha} h\bigl(z+(1-\tau)u\bigr) \\&\quad\times \Upsilon_\tau\bigl(z-\tau u,z+(1-\tau)u\bigr)^{-1} {(\i\partial_u)}_{\boldsymbol\alpha} \e^{-\i u\cdot p} \d z \d u \,\frac{\d p}{(2\uppi)^d}
    \end{split} \\
    \begin{split}
      &= \int_{M\times \T_zM \times \T^*_zM} \e^{-\i u\cdot p}{(-\i\partial_u)}_{\boldsymbol\alpha} \bar{f(z-\tau u)} a(z)^{\boldsymbol\alpha} h\bigl(z+(1-\tau)u\bigr) \\&\quad\times \Upsilon_\tau\bigl(z-\tau u,z+(1-\tau)u\bigr)^{-1} \d z \d u \,\frac{\d p}{(2\uppi)^d}
    \end{split} \\
    \begin{split}
      &= \int_M {(-\i\partial_u)}_{\boldsymbol\alpha} \bar{f(z-\tau u)} a(z)^{\boldsymbol\alpha} h\bigl(z+(1-\tau)u\bigr) \\&\quad\times \Upsilon_\tau\bigl(z-\tau u,z+(1-\tau)u\bigr)^{-1} \d z \,\Big|_{u=0}.
    \end{split}
  \end{align*}
  After an integration by parts, this implies~\eqref{eq:monomial_quant}.
\end{proof}

For the balanced geodesic Weyl quantization, \eqref{eq:monomial_quant} simplifies and we find:
\begin{proposition}
  Consider a polynomial symbol~\eqref{eq:polypo}, as in Prop.~\ref{prop:monomial_quant}.
  Then
  \begin{equation*}
    \Op(a) = 2^{-|\boldsymbol\alpha+\boldsymbol\beta+\boldsymbol\gamma|} \hbar^{|\boldsymbol\gamma|} \abs{g}^{-\frac14} (-\i\hbar\nabla)_{\boldsymbol\alpha} a^{\boldsymbol\alpha+\boldsymbol\beta+\boldsymbol\gamma} \xi_{\boldsymbol\gamma} \abs{g}^\frac12 (-\i\hbar\nabla)_{\boldsymbol\beta} \abs{g}^{-\frac14}
  \end{equation*}
  with
  \begin{equation}\label{eq:xi_weyl}
    \xi_{\boldsymbol\gamma} \defn \sum_{\mathclap{\boldsymbol\alpha+\boldsymbol\beta=\boldsymbol\gamma}} (-1)^{\abs{\boldsymbol\alpha}} {(-\i\nabla)}_{\boldsymbol\alpha} {(-\i\nabla')}_{\boldsymbol\beta} \Delta(x,y)^{-\frac12} \Big|_{x=y}.
  \end{equation}
  Moreover, if $a$ is an even polynomial, $\Op(a)$ has only even degree derivatives, and if $a$ is an odd polynomial, $\Op(a)$ has only odd degree derivatives.
\end{proposition}
\begin{proof}
  The second part of the proposition follows from
  \begin{equation*}
    {(-\i\nabla)}_{\boldsymbol\alpha} {(-\i\nabla')}_{\boldsymbol\beta} \Delta(x,y)^{-\frac12} \Big|_{x=y} = {(-\i\nabla)}_{\boldsymbol\beta} {(-\i\nabla')}_{\boldsymbol\alpha} \Delta(x,y)^{-\frac12} \Big|_{x=y}
  \end{equation*}
  due to the symmetry of the Van Vleck--Morette determinant $\Delta(x,y)=\Delta(y,x)$.
  Indeed, for odd $\abs{\boldsymbol\gamma}$, either $\abs{\boldsymbol\alpha}$ or $\abs{\boldsymbol\beta}$ is odd and thus the sum~\eqref{eq:xi_weyl} consists of terms of the form
  \begin{align*}
    0 &= (-1)^{\abs{\boldsymbol\alpha}} {(-\i\nabla)}_{\boldsymbol\alpha} {(-\i\nabla')}_{\boldsymbol\beta} \Delta(x,y)^{-\frac12} \\&\quad + (-1)^{\abs{\boldsymbol\beta}} {(-\i\nabla)}_{\boldsymbol\beta} {(-\i\nabla')}_{\boldsymbol\alpha} \Delta(x,y)^{-\frac12} \Big|_{x=y}.
    \qedhere
  \end{align*}
\end{proof}

Consider the example of a quadratic symbol $a(z,p) = a(z)^{\mu\nu} p_\mu p_\nu$.
Then
\begin{align*}
  \frac{1}{\hbar^2} \Op_\tau(a) &=
  -\abs{g}^{-\frac14} \bigl(\tau^2 \nabla_\mu \nabla_\nu a^{\mu\nu} \abs{g}^{\frac12} + 2\tau (1-\tau) \nabla_\mu a^{\mu\nu} \abs{g}^{\frac12} \nabla_\nu \\&\quad + (1-\tau)^2 a^{\mu\nu} \abs{g}^{\frac12} \nabla_\mu \nabla_\nu \bigr) \abs{g}^{-\frac14} + \frac16 a^{\mu\nu} R_{\mu\nu},
\end{align*}
where we used the fact that
\begin{equation*}
  \nabla_\mu \Delta(x,y)^{-\frac12} \Big|_{x=y} = 0,
  \quad
  \nabla_\mu \nabla_\nu \Delta(x,y)^{-\frac12} \Big|_{x=y} = -\frac16 R_{\mu\nu}.
\end{equation*}
Suppose that $a^{\mu\nu} = g^{\mu\nu}$ is the inverse metric.
Since the metric is covariantly constant, we obtain
\begin{equation*}
  \frac{1}{\hbar^2} \Op_\tau(a)
  = \abs{g}^{-\frac14} (-\i\nabla_\mu) g^{\mu\nu} \abs{g}^\frac12 (-\i\nabla_\nu) \abs{g}^{-\frac14} + \frac16 R
  = -g^{\mu\nu} \nabla_\mu \nabla_\nu + \frac16 R
\end{equation*}
independently of $\tau$.
In four dimensions this is the conformally invariant Laplace--Beltrami operator (or its pseudo-Riemannian generalization).

\subsection{Balanced geodesic star product}

In this section, we consider only the balanced geodesic Weyl quantization.
We define the \emph{balanced geodesic star product} by the identity
\begin{equation*}
  \Op(a \star b) = \Op(a) \Op(b).
\end{equation*}
Let us first describe an explicit formula for the star product, which generalizes~\eqref{eq:moyal} from the flat case.
For simplicity, we assume that we can take $\chi=1$.

\begin{figure}
  \centering
  \begin{tikzpicture}[scale=1.3]
    \fill[fill=red!20] (0,0)--(3,0)--(1.75,1.5)--cycle;
    \fill[fill=blue!20] (6,0)--(3,0)--(4.75,1.5)--cycle;
    \fill[fill=green!20] (3.5,3)--(3,0)--(1.75,1.5)--cycle;
    \draw[fill=yellow!20] (3.5,3)--(3,0)--(4.75,1.5)--cycle;

    \fill (3,0) circle [radius=2pt] node[below=2pt] {$z$};
    \draw (3,0) circle [radius=3pt];
    \fill (3.5,3) circle [radius=2pt] node[above] {$\tilde z$};
    \fill (0,0) circle [radius=2pt] node[below left=-1pt] {$x$};
    \fill (6,0) circle [radius=2pt] node[below right=-1pt] {$y$};
    \fill (1.75,1.5) circle [radius=2pt] node[above left=-1pt] {$z_1$};
    \fill (4.75,1.5) circle [radius=2pt] node[above right=-1pt] {$z_2$};

    \draw[->,shorten >=3pt,thick] (4.75,1.5) -- node[above right=-1pt] {$\parap{u_1}^{v_2}$} (3.5,3);
    \draw[->,shorten >=3pt,thick] (1.75,1.5) -- node[above left=-1pt] {$\parap{u_2}^{\mathrlap{v_1}}$} (3.5,3);
    \draw[->,shorten >=3pt,thick] (3,0) -- node[left=4pt] {$u_1$} (1.6,1.4);
    \draw[->,shorten >=3pt,thick] (3,0) -- node[right=4pt] {$u_2$} (4.8,1.3);
    \draw[->,shorten >=3pt,thick] (3,0) -- node[above=4pt] {$v_1$} (1.75,1.5);
    \draw[->,shorten >=3pt,thick] (3,0) -- node[above=4pt] {$v_2$} (4.75,1.5);
    \draw[->,shorten >=3pt,thick] (3,0) -- node[left] {$\tilde w$} (3.5,3);
    \draw[->,shorten >=3pt,thick] (3,0) -- node[below] {$w$} (6,0);

    \draw[dashed] (3,0) -- (0,0);
    \draw[dashed] (1.75,1.5) -- (0,0);
    \draw[dashed] (4.75,1.5) -- (6,0);
  \end{tikzpicture}
  \caption{\label{fig:star-product-triangle}%
  The geodesic triangle in the proof of Thm.~\ref{thm:star-product}.
  $z, z_1, z_2$ are the middle-points of the three geodesics between the points $x, y, \tilde{z}$ spanning the triangle.
  $w$ is the tangent to the geodesic from $z$ to $y$; $\tilde{w}$ is the tangent to the geodesic from $z$ to $\tilde{z}$; $v_1$ and $v_2$ are the tangents to the geodesics from $z$ to $z_1$ resp.\ $z_2$.
  $u_1$ and $u_2$ are the parallel transports of the tangents to the geodesics from $z_2$ resp.\ $z_1$ to $\tilde{z}$ along $v_1$ resp.\ $v_2$.
  In the flat case, $u_i = v_i$.
  }
\end{figure}
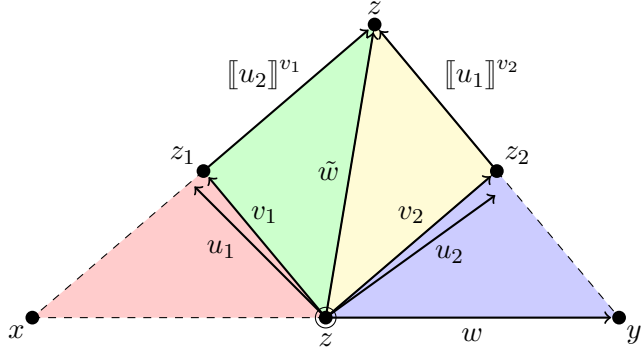

\begin{theorem}\label{thm:star-product}
  The balanced geodesic star product is given by the formulas
  \begin{subequations}\begin{align}\MoveEqLeft
    (a \star b)(z,p) \nonumber \\
    \begin{split}\label{eq:geom-star-prod-1}
      &= \int_{\T_zM \times \T_zM \times \T_z^*M \times \T_z^*M} a(z+v_1,\parap{p+p_1}^{v_1}) b(z+v_2,\parap{p+p_2}^{v_2}) \\&\quad\times \Lambda(z,u_1,u_2) \e^{\frac{2\i (w+u_1-u_2) \cdot p}{\hbar}} \e^{\frac{2\i (u_1 \cdot p_2 - u_2 \cdot p_1)}{\hbar}} \d u_1 \d u_2 \frac{\d p_1 \d p_2}{(\uppi\hbar)^{2d}},
    \end{split} \\
    \begin{split}\label{eq:geom-star-prod-2}
      &\sim \exp\bigl(\tfrac{\i}{2}\hbar(\partial_{u_1} \cdot \partial_{p_2} - \partial_{u_2} \cdot \partial_{p_1})\bigr) \Lambda(z,u_1,u_2) \e^{\frac{2\i (w+u_1-u_2) \cdot p}{\hbar}} \\&\quad\times a(z+v_1,\parap{p+p_1}^{v_1}) b(z+v_2,\parap{p+p_2}^{v_2}) \Big|_{\substack{u_1=u_2=0\\p_1=p_2=0}},
    \end{split}
  \end{align}\end{subequations}
  where
  \begin{equation}\label{eq:Lambda-factor}
    \Lambda(z,u_1,u_2) \defn 2^{-d} \abs*{\frac{\partial(w,\tilde{w})}{\partial(u_1,u_2)}} \frac{\Delta(z-w,z+\tilde{w})^\frac12 \Delta(z+w,z+\tilde{w})^\frac12}{\Delta(z-w,z+w)^\frac12 \Delta(z,z+\tilde{w})},
  \end{equation}
  and the vectors $u_1, u_2, v_1, v_2, w, \tilde{w} \in \T_zM$ satisfy the relations
  \begin{subequations}\label{eq:triangle-points-vectors}\begin{align}
    z - w &= (z + v_1) - \parap{u_2}^{v_1}, \\
    z + w &= (z + v_2) - \parap{u_1}^{v_2}, \\
    z + \tilde{w} &= (z + v_1) + \parap{u_2}^{v_1} = (z + v_2) + \parap{u_1}^{v_2}.
  \end{align}\end{subequations}
  The schematic arrangement of these vectors can be seen in Fig.~\ref{fig:star-product-triangle}.
\end{theorem}
\begin{proof}
  Let $C(x,y)$ be the integral kernel of $\Op(a\star b)$.
  Clearly,
  \begin{equation}\label{pqpq}
    C(x,y)=\int_M\Op(a)(x,\tilde z)\Op(b)(\tilde z,y)\rd \tilde z.
  \end{equation}
  Let $z$ be the middle point between~$x$ and~$y$.
  Let $w,\tilde w \in \T_zM$ be defined by
  \begin{equation*}
    x = z-w,
    \quad
    y = z+w,
    \quad
    \tilde{z} = z+\tilde{w}.
  \end{equation*}
  Then we can rewrite~\eqref{pqpq} as
  \begin{align*}\MoveEqLeft
    C(z-w,z+w) \\
    &\mathrel{\mathrlap{\phantom\defn}=} \int_M \Op(a)(z-w,\tilde{z}) \Op(b)(\tilde{z},z+w) \d\tilde{z} \\
    &\mathrel{\mathrlap{\phantom\defn}=} \int_{\T_zM} \Upsilon_0(z,z+\tilde{w})^{-2} \Op(a)(z-w,z+\tilde{w}) \Op(b)(z+\tilde{w},z+w) \d\tilde{w}.
  \end{align*}
  The star product of the symbols can then be found by applying~\eqref{eq:kernel-to-symbol}:
  \begin{align}
    (a \star b)(z,p)
    &= 2^d \int_{\T_zM} \Upsilon_\frac12(z-w,z+w)^{-1} C(z-w,z+w) \e^{2\i w \cdot p} \d w \nonumber \\
    \begin{split}\label{eq:star-prod-intermediate}
      &= 2^d \int_{\T_zM \times \T_zM} \Upsilon_\frac12(z-w,z+w)^{-1} \Upsilon_0(z,z+\tilde{w})^{-2} \\&\quad\times \Op(a)(z-w,z+\tilde{w}) \Op(b)(z+\tilde{w},z+w) \e^{2\i w \cdot p} \d w \d\tilde{w}.
    \end{split}
  \end{align}

  Consider next the vectors $u_1, u_2, v_1, v_2 \in \T_zM$ defined in~\eqref{eq:triangle-points-vectors}.
  $z+v_1$ is the middle point between $z-w$ and $z+\tilde w$, and $z+v_2$ is the middle point between $z+w$ and $z+\tilde w$.
  Therefore,
  \begin{align*}
    \begin{split}\MoveEqLeft
      \Op(a)(z-w,z+\tilde{w}) \\
      &= \Upsilon_\frac12(z-w,z+\tilde{w}) \int_{\T_{z_1}^*M} a(z+v_1,\parap{q_1}^{v_1}) \e^{-2\i u_2 \cdot q_1} \frac{\d\parap{q_1}^{v_1}}{(2\uppi)^d} \\
      &= \Upsilon_\frac12(z-w,z+\tilde{w}) \frac{\abs{g(z+v_1)}^\frac12}{\abs{g(z)}^\frac12} \int_{\T_z^*M} a(z+v_1,\parap{q_1}^{v_1}) \e^{-2\i u_2 \cdot q_1} \frac{\d q_1}{(2\uppi)^d},
    \end{split}
    \\
    \begin{split}\MoveEqLeft
      \Op(b)(z+\tilde{w},z+w) \\
      &= \Upsilon_\frac12(z+w,z+\tilde{w}) \frac{\abs{g(z+v_2)}^\frac12}{\abs{g(z)}^\frac12} \int_{\T_z^*M} b(z+v_2,\parap{q_2}^{v_2}) \e^{2\i u_1 \cdot q_2} \frac{\d q_2}{(2\uppi)^d}.
    \end{split}
  \end{align*}
  Changing the integration variables in~\eqref{eq:star-prod-intermediate} from $w,\tilde{w}$ to $u_1,u_2$, and inserting the formulas for $\Op(a)$, $\Op(b)$, we obtain
  \begin{align*}
    \begin{split}\eqref{eq:star-prod-intermediate}
      &= 2^d \int_{\T_zM \times \T_zM} \abs*{\frac{\partial(w,\tilde{w})}{\partial(u_1,u_2)}} \Upsilon_\frac12(z-w,z+w)^{-1} \Upsilon_0(z,z+\tilde{w})^{-2} \\&\quad\times \Op(a)(z-w,z+\tilde w) \Op(b)(z+\tilde w,z+w) \e^{2\i w \cdot p} \d u_1 \d u_2
    \end{split}
    \\
    \begin{split}
      &= \int_{\T_zM \times \T_zM \times \T_z^*M \times \T_z^*M} a(z+v_1,\parap{q_1}^{v_1}) b(z+v_2,\parap{q_2}^{v_2}) \\&\quad\times \Lambda(z,u_1,u_2) \e^{2\i (w \cdot p + u_1 \cdot q_2 - u_2 \cdot q_1)} \d u_1 \d u_2 \frac{\d q_1 \d q_2}{\uppi^{2d}}
    \end{split}
    \\
    \begin{split}
      &= \int_{\T_zM \times \T_zM \times \T_z^*M \times \T_z^*M} a(z+v_1,\parap{p+p_1}^{v_1}) b(z+v_2,\parap{p+p_2}^{v_2}) \\&\quad\times \Lambda(z,u_1,u_2) \e^{2\i (w+u_1-u_2) \cdot p} \e^{2\i (u_1 \cdot p_2 - u_2 \cdot p_1)} \d u_1 \d u_2 \frac{\d p_1 \d p_2}{\uppi^{2d}},
    \end{split}
  \end{align*}
  with $\Lambda$ as defined in~\eqref{eq:Lambda-factor}.

  This yields~\eqref{eq:geom-star-prod-1}.
  Using~\eqref{eq:gaussian}, we obtain~\eqref{eq:geom-star-prod-2}.
\end{proof}

\begin{remark}
  A similar formula can be given also for $\tau \neq \frac12$.
  For this purpose, make in~\eqref{eq:Lambda-factor} the substitutions
  \begin{align}
    z - w \mapsto z - 2\tau w, \quad
    z + w \mapsto z + 2(1-\tau) w,
  \end{align}
  and replace~\eqref{eq:triangle-points-vectors} by
  \begin{align*}
    z - 2\tau w &= (z + v_1) - 2\tau \parap{u_2}^{v_1}, \\
    z + 2(1-\tau) w &= (z + v_2) - 2(1-\tau) \parap{u_1}^{v_2}, \\
    z + \tilde{w} &= (z + v_1) + 2(1-\tau) \parap{u_2}^{v_1} = (z + v_2) + 2\tau \parap{u_1}^{v_2}.
  \end{align*}
  Unlike in the flat case, all derivatives in~\eqref{eq:geom-star-prod-2} remain also for $\tau=0$ and $\tau=1$ because of the non-trivial geometric factor $\Lambda(z,u_1,u_2) \exp(2\i(w+u_1+u_2)\cdot p)$ and the dependence of $v_1$ and $v_2$ on both $u_1$ and $u_2$.
  Note, however, that $v_1=0$ for $\tau=0$ and $v_2=0$ for $\tau=1$.
\end{remark}

\subsection{Asymptotic expansion of the geodesic star product}

The geodesic star product can be expanded as a sum
\begin{equation}\label{eq:starprod-expansion}
  a \starprod b = \sum_n \hbar^n(a \starprod b)_n + O(\hbar^\infty)
\end{equation}
according to the order of Planck's constant.
Note that each $n$-th term contains exactly $n$ position derivatives, with Riemann tensors counting as two derivatives.
It also contains exactly $n$ momentum derivatives, with multiplication by~$p$ counting as $-1$ derivatives.
The asymptotic expansion does not depend on the cutoff $\chi$.

Due to the length of the expressions involved, we shall adopt the following notation for the derivatives of the symbols:
lower indices always denote horizontal derivatives
\begin{equation*}
  a_{\alpha_1 \dotsm \alpha_n} = \nnabla_{\alpha_n} \dotsm \nnabla_{\alpha_1} a,
\end{equation*}
and upper indices always denote vertical derivatives
\begin{equation*}
  a^{\alpha_1 \dotsm \alpha_n} = \ppartial^{\alpha_n} \dotsm \ppartial^{\alpha_1} a.
\end{equation*}
Since we only consider scalar symbols, no ambiguity arises.
Recall also that horizontal and vertical derivatives commute so that their relative position is irrelevant.

In this notation, the five lowest order summands in~\eqref{eq:starprod-expansion} are
\begin{align*}
  (a \starprod b)_0 &= a b,
  \\
  (a \starprod b)_1 &=
  \tfrac{\i}{2} \bigl(a_{\alpha} b^{\alpha} - a^{\alpha} b_{\alpha}\bigr),
  \\
  \begin{split}
    (a \starprod b)_2 &=
    -\tfrac18 \bigl(a_{\alpha_1 \alpha_2} b^{\alpha_1 \alpha_2} - 2 a^{\alpha_2}_{\alpha_1} b^{\alpha_1}_{\alpha_2} + a^{\alpha_1 \alpha_2} b_{\alpha_1 \alpha_2}\bigr)
    + \tfrac{1}{12} R_{\alpha_1 \alpha_2} a^{\alpha_2} b^{\alpha_1}
    \\&\quad
    - \tfrac{1}{24} R^\beta{}_{\alpha_1 \alpha_2 \alpha_3} p_\beta \bigl(a^{\alpha_2} b^{\alpha_1 \alpha_3} + a^{\alpha_1 \alpha_3} b^{\alpha_2}\bigr),
  \end{split}
  \\
  \begin{split}
    (a \starprod b)_3 &=
    -\tfrac{\i}{48} \bigl( a_{\alpha_1 \alpha_2 \alpha_3} b^{\alpha_1 \alpha_2 \alpha_3} - 3 a^{\alpha_3}_{\alpha_1 \alpha_2} b^{\alpha_1 \alpha_2}_{\alpha_3} + 3 a^{\alpha_2 \alpha_3}_{\alpha_1} b^{\alpha_1}_{\alpha_2 \alpha_3} - a^{\alpha_1 \alpha_2 \alpha_3} b_{\alpha_1 \alpha_2 \alpha_3} \bigr)
    \\&\quad
    + \tfrac{\i}{24} R_{\alpha_1 \alpha_2} \bigl( a^{\alpha_2}_{\alpha_3} b^{\alpha_1 \alpha_3} - a^{\alpha_2 \alpha_3} b^{\alpha_1}_{\alpha_3} \bigr)
    - \tfrac{\i}{16} R^\beta{}_{\alpha_1 \alpha_2 \alpha_3} \bigl(a^{\alpha_1 \alpha_3}_\beta b^{\alpha_2} - a^{\alpha_2} b^{\alpha_1 \alpha_3}_\beta \bigr)
    \\&\quad
    - \tfrac{\i}{48} R^\beta{}_{\alpha_1 \alpha_2 \alpha_3} p_\beta \bigl( -a^{\alpha_1 \alpha_3}_{\alpha_4} b^{\alpha_2 \alpha_4} - a^{\alpha_2}_{\alpha_4} b^{\alpha_1 \alpha_3 \alpha_4} + a^{\alpha_1 \alpha_3 \alpha_4} b^{\alpha_2}_{\alpha_4} + a^{\alpha_2 \alpha_4} b^{\alpha_1 \alpha_3}_{\alpha_4} \bigr)
    \\&\quad
    + \tfrac{\i}{48} R_{\alpha_1 \alpha_2 ; \alpha_3} \bigl( a^{\alpha_3} b^{\alpha_1 \alpha_2} - a^{\alpha_1 \alpha_2} b^{\alpha_3} \bigr)
    \\&\quad
    + \tfrac{\i}{48} R^\beta{}_{\alpha_1 \alpha_2 \alpha_3 ; \alpha_4} p_\beta \bigl( a^{\alpha_1 \alpha_3 \alpha_4} b^{\alpha_2} - a^{\alpha_2} b^{\alpha_1 \alpha_3 \alpha_4} \bigr),
  \end{split}
  \\
  \begin{split}
    (a \starprod b)_4 &=
    \tfrac{1}{384} \bigl( a_{\alpha_1 \alpha_2 \alpha_3 \alpha_4} b^{\alpha_1 \alpha_2 \alpha_3 \alpha_4} - 4 a^{\alpha_4}_{\alpha_1 \alpha_2 \alpha_3} b^{\alpha_1 \alpha_2 \alpha_3}_{\alpha_4} + 6 a_{\alpha_1 \alpha_2}^{\alpha_3 \alpha_4} b_{\alpha_3 \alpha_4}^{\alpha_1 \alpha_2} \\&\qquad - 4 a_{\alpha_1}^{\alpha_2 \alpha_3 \alpha_4} b_{\alpha_2 \alpha_3 \alpha_4}^{\alpha_1} + a^{\alpha_1 \alpha_2 \alpha_3 \alpha_4} b_{\alpha_1 \alpha_2 \alpha_3 \alpha_4} \bigr)
    \\&\quad
    - \tfrac{1}{96} R_{\alpha_1 \alpha_2} \bigl( a^{\alpha_2}_{\alpha_3 \alpha_4} b^{\alpha_1 \alpha_3 \alpha_4} - 2 a^{\alpha_2 \alpha_4}_{\alpha_3} b^{\alpha_1 \alpha_3}_{\alpha_4} + a^{\alpha_2 \alpha_3 \alpha_4} b^{\alpha_1}_{\alpha_3 \alpha_4} \bigr)
    \\&\quad
    + \tfrac{1}{32} R^\beta{}_{\alpha_1 \alpha_2 \alpha_3} \bigl( a^{\alpha_1 \alpha_3}_{\beta\alpha_4} b^{\alpha_2 \alpha_4} - a^{\alpha_1 \alpha_3 \alpha_4}_\beta b^{\alpha_2}_{\alpha_4} - a^{\alpha_2}_{\alpha_4} b^{\alpha_1 \alpha_3 \alpha_4}_\beta + a^{\alpha_2 \alpha_4} b^{\alpha_1 \alpha_3}_{\beta\alpha_4} \bigr)
    \\&\quad
    + \tfrac{1}{192} R^\beta{}_{\alpha_1 \alpha_2 \alpha_3} p_\beta \bigl( a^{\alpha_1 \alpha_3}_{\alpha_4 \alpha_5} b^{\alpha_2 \alpha_4 \alpha_5} + a^{\alpha_2}_{\alpha_4 \alpha_5} b^{\alpha_1 \alpha_3 \alpha_4 \alpha_5} - 2 a^{\alpha_1 \alpha_3 \alpha_5}_{\alpha_4} b^{\alpha_2 \alpha_4}_{\alpha_5} \\&\qquad - 2 a^{\alpha_2 \alpha_5}_{\alpha_4} b^{\alpha_1 \alpha_3 \alpha_4}_{\alpha_5} + a^{\alpha_1 \alpha_3 \alpha_4 \alpha_5} b^{\alpha_2}_{\alpha_4 \alpha_5} + a^{\alpha_2 \alpha_4 \alpha_5} b^{\alpha_1 \alpha_3}_{\alpha_4 \alpha_5} \bigr)
    \\&\quad
    - \tfrac{1}{96} R_{\alpha_1 \alpha_2 ; \alpha_3} \bigl( a^{\alpha_3}_{\alpha_4} b^{\alpha_1 \alpha_2 \alpha_4} - a^{\alpha_3 \alpha_4} b^{\alpha_1 \alpha_2}_{\alpha_4} - a^{\alpha_1 \alpha_2}_{\alpha_4} b^{\alpha_3 \alpha_4} + a^{\alpha_1 \alpha_2 \alpha_4} b^{\alpha_3}_{\alpha_4} \bigr)
    \\&\quad
    - \tfrac{1}{384} R^\beta{}_{\alpha_1 \alpha_2 \alpha_3 ; \alpha_4} \bigl( 2 a^{\alpha_1 \alpha_3 \alpha_4}_\beta b^{\alpha_2} + a^{\alpha_2 \alpha_4}_\beta b^{\alpha_1 \alpha_3} - 5 a^{\alpha_1 \alpha_3}_\beta b^{\alpha_2 \alpha_4} \\&\qquad - 2 a^{\alpha_2}_\beta b^{\alpha_1 \alpha_3 \alpha_4} + 2 a^{\alpha_2} b^{\alpha_1 \alpha_3 \alpha_4}_\beta - 5 a^{\alpha_2 \alpha_4} b^{\alpha_1 \alpha_3}_\beta + a^{\alpha_1 \alpha_3} b^{\alpha_2 \alpha_4}_\beta \\&\qquad - 2 a^{\alpha_1 \alpha_3 \alpha_4} b^{\alpha_2}_\beta \bigr)
    \\&\quad
    - \tfrac{1}{96} R^\beta{}_{\alpha_1 \alpha_2 \alpha_3 ; \alpha_4} p_\beta \bigl( a^{\alpha_1 \alpha_3 \alpha_4}_{\alpha_5} b^{\alpha_2 \alpha_5} - a^{\alpha_2}_{\alpha_5} b^{\alpha_1 \alpha_3 \alpha_4 \alpha_5} - a^{\alpha_1 \alpha_3 \alpha_4 \alpha_5} b^{\alpha_2}_{\alpha_5} \\&\qquad + a^{\alpha_2 \alpha_5} b^{\alpha_1 \alpha_3 \alpha_4}_{\alpha_5} \bigr)
    \\&\quad
    + \tfrac{1}{288} R_{\alpha_1 \alpha_2} R_{\alpha_3 \alpha_4} a^{\alpha_2 \alpha_4} b^{\alpha_1 \alpha_3}
    \\&\quad
    - \tfrac{1}{288} R_{\alpha_1 \alpha_2} R^\beta{}_{\alpha_3 \alpha_4 \alpha_5} p_\beta \bigl(a^{\alpha_1 \alpha_3 \alpha_5} b^{\alpha_2 \alpha_4} + a^{\alpha_2 \alpha_4} b^{\alpha_1 \alpha_3 \alpha_5} \bigr)
    \\&\quad
    - \tfrac{1}{288} R_{\beta \alpha_1} R^\beta{}_{\alpha_2 \alpha_3 \alpha_4} \bigl(a^{\alpha_2 \alpha_4} b^{\alpha_1 \alpha_3}+ a^{\alpha_1 \alpha_3} b^{\alpha_2 \alpha_4} + a^{\alpha_3} b^{\alpha_1 \alpha_2 \alpha_4} \\&\qquad + a^{\alpha_1 \alpha_2 \alpha_4} b^{\alpha_3} \bigr)
    \\&\quad
    - \tfrac{1}{2880} R^\beta{}_{\alpha_1}{}^\gamma{}_{\alpha_2} R_{\beta \alpha_3 \gamma \alpha_4} \bigl( 28 a^{\alpha_2 \alpha_3 \alpha_4} b^{\alpha_1} + 14 a^{\alpha_1 \alpha_2} b^{\alpha_3 \alpha_4} - 31 a^{\alpha_1 \alpha_3} b^{\alpha_2 \alpha_4} \\&\qquad + 14 a^{\alpha_1 \alpha_4} b^{\alpha_2 \alpha_3} + 28 a^{\alpha_1} b^{\alpha_2 \alpha_3 \alpha_4} \bigr)
    \\&\quad
    + \tfrac{1}{5760} R^\beta{}_{\alpha_1 \gamma\alpha_2} R^\gamma{}_{\alpha_3 \alpha_4 \alpha_5} p_\beta \bigl(7 a^{\alpha_1 \alpha_2 \alpha_3 \alpha_5} b^{\alpha_4} - 11 a^{\alpha_1 \alpha_3 \alpha_5} b^{\alpha_2 \alpha_4} \\&\qquad - 44 a^{\alpha_2 \alpha_3 \alpha_5} b^{\alpha_1 \alpha_4} + 7 a^{\alpha_1 \alpha_2 \alpha_4} b^{\alpha_3 \alpha_5} + 7 a^{\alpha_3 \alpha_5} b^{\alpha_1 \alpha_2 \alpha_4} \\&\qquad - 44 a^{\alpha_1 \alpha_4} b^{\alpha_2 \alpha_3 \alpha_5} - 11 a^{\alpha_2 \alpha_4} b^{\alpha_1 \alpha_3 \alpha_5} + 7 a^{\alpha_4} b^{\alpha_1 \alpha_2 \alpha_3 \alpha_5} \bigr)
    \\&\quad
    + \tfrac{1}{1152} R^\beta{}_{\alpha_1 \alpha_2 \alpha_3} R^\gamma{}_{\alpha_4 \alpha_5 \alpha_6} p_\beta p_\gamma \bigl(a^{\alpha_1 \alpha_3 \alpha_4 \alpha_6} b^{\alpha_2 \alpha_5} + 2a^{\alpha_1 \alpha_3 \alpha_5} b^{\alpha_2 \alpha_4 \alpha_6} \\&\qquad + a^{\alpha_2 \alpha_5} b^{\alpha_1 \alpha_3 \alpha_4 \alpha_6}\bigr)
    \\&\quad
    - \tfrac{1}{1920} R_{\alpha_1 \alpha_2 ; \alpha_3 \alpha_4} \bigl( a^{\alpha_1 \alpha_2 \alpha_3} b^{\alpha_4} + a^{\alpha_1 \alpha_2 \alpha_4} b^{\alpha_3} + 2 a^{\alpha_1 \alpha_3 \alpha_4} b^{\alpha_2} \\&\qquad + 2 a^{\alpha_1 \alpha_3} b^{\alpha_2 \alpha_4} - 9 a^{\alpha_1 \alpha_2} b^{\alpha_3 \alpha_4} - 9 a^{\alpha_3 \alpha_4} b^{\alpha_1 \alpha_2} + 2 a^{\alpha_2 \alpha_4} b^{\alpha_1 \alpha_3} \\&\qquad + 2 a^{\alpha_2} b^{\alpha_1 \alpha_3 \alpha_4} + a^{\alpha_3} b^{\alpha_1 \alpha_2 \alpha_4} + a^{\alpha_4} b^{\alpha_1 \alpha_2 \alpha_3} \bigr)
    \\&\quad
    + \tfrac{1}{1920} R^\beta{}_{\alpha_1 \alpha_2 \alpha_3 ;\alpha_4 \alpha_5}p_\beta \bigl(3 a^{\alpha_1 \alpha_3 \alpha_4 \alpha_5} b^{\alpha_2} + 3a^{\alpha_2} b^{\alpha_1 \alpha_3 \alpha_4 \alpha_5} \\&\qquad + a^{\alpha_2 \alpha_4 \alpha_5} b^{\alpha_1 \alpha_3} + a^{\alpha_1 \alpha_3} b^{\alpha_2 \alpha_4 \alpha_5} - 7 a^{\alpha_1 \alpha_3 \alpha_4} b^{\alpha_2 \alpha_5} - 7 a^{\alpha_2 \alpha_5} b^{\alpha_1 \alpha_3 \alpha_4} \bigr).
  \end{split}
\end{align*}

\subsection{Analysis of the expansion of the star product}

In this subsection we analyze terms that appear in the expansion of the star product.
The term that appears at $\hbar^r$, that is $(a\star b)_r(z,p)$ is a linear combination with numerical coefficients of terms of the form
\begin{equation*}
  \Bigl( \prod_{j=1}^s R(z)^{\pi_j}{}_{\boldsymbol{\rho}_j;\boldsymbol{\nu}_j} \Bigr)
  p_{\boldsymbol{\eta}}
  a(z,p)_{\boldsymbol{\alpha}_1}^{\boldsymbol{\beta}_1}
  b(z,p)_{\boldsymbol{\alpha}_2}^{\boldsymbol{\beta}_2},
\end{equation*}
where we use the same notation for the derivatives of the symbols $a$ and $b$ as in the previous subsection.

We have $|\bs{\rho}_i|=3$, $\pi_i$ are single indices, $i=1,\dots s$.
The following identities are always satisfied:
\begin{align*}
  r &= |\boldsymbol\beta_1|+|\boldsymbol\beta_2|-|\boldsymbol\eta| \\
  r &= 2s+\sum_{i=1}^s|\boldsymbol\nu_i|+|\boldsymbol\alpha_1|+|\boldsymbol\alpha_2|.
\end{align*}
Observe that the first identity can be read off immediately from~\eqref{eq:geom-star-prod-2} and then the second follows by the fact that all indices must be contracted.

These terms can be divided into two kinds:
\begin{enumerate}
  \item
    Terms that have the same form as the terms appearing in the star product on a flat space.
    They satisfy
    \begin{equation*}
      |\boldsymbol\eta|=s=\sum_{j=1}^s|\boldsymbol\nu_j|=0,\quad |\boldsymbol\beta_1|=|\boldsymbol\alpha_2|,\quad
      |\boldsymbol\beta_2|=|\boldsymbol\alpha_1|.
    \end{equation*}
  \item
    Terms that contain the curvature tensor and its covariant derivatives.
    They satisfy
    \begin{align*}
      s &\geq \max(1,|\boldsymbol{\eta}|), \\
      |\boldsymbol{\beta_1}| &\geq \max(1,|\boldsymbol{\eta}|+|\boldsymbol{\alpha_2}|), \\
      |\boldsymbol{\beta_2}| &\geq \max(1,|\boldsymbol{\eta}|+|\boldsymbol{\alpha_1}|).
    \end{align*}
\end{enumerate}

For the star product of several symbols we also can write
\begin{equation*}
  (a_1\star \dotsm \star a_n)(z,p) = \sum_{r=0}^\infty \hbar^r (a_1\star \dotsm \star a_n)_r(z,p) +O(\hbar^\infty),
\end{equation*}
where $(a_1\star \dotsm \star a_n)_r(z,p)$ is a linear combination with numerical coefficients of terms of the form
\begin{equation*}
  \Bigl( \prod_{j=1}^s R(z)_{\boldsymbol{\rho}_j;\boldsymbol{\nu}_j}^{\pi_j} \Bigr)
  p_{\boldsymbol{\eta}}
  a_1(z,p)_{\boldsymbol{\alpha}_1}^{\boldsymbol{\beta}_1} \dotsm a_n(z,p)_{\boldsymbol{\alpha}_n}^{\boldsymbol{\beta}_n}.
\end{equation*}
The identities and bounds for $n=2$ generalize to this case:
\begin{align*}
  r &= 2s+\sum_{i=1}^s|\boldsymbol{\nu}_i|+\sum_{j=1}^n|\boldsymbol{\alpha}_j| = \sum_{j=1}^n|\boldsymbol{\beta}_j|-|\bs{\eta}|, \\
  s &\geq |\boldsymbol{\eta}|.
\end{align*}

\section{Calculation of the star product}
\label{Methods}

In this section we explain the methods that we used to obtain the expansion of the star product.
We also give various intermediate results.

\subsection{Coincidence limits of Synge's world function}
\label{sec:synge-coinc}

In our calculations, we will need a certain family of coincidence limits of covariant derivatives of the Synge's function.
More precisely, we will need
\begin{equation*}
  [\sigma^{\mu}{}_{(\boldsymbol\alpha')(\boldsymbol\beta')}] = [\sigma^{\mu}{}_{(\alpha_1'\dotsm\alpha_m')(\beta_1'\dotsm\beta_n')}],
\end{equation*}
where parentheses around indices indicate symmetrization.

For notational simplicity, below we shall use the same index repeatedly to indicate symmetrization, e.g.
\begin{equation*}
  R^\mu{}_{\alpha\beta\alpha} = \frac12 (R^\mu{}_{\alpha_1\beta\alpha_2} + R^\mu{}_{\alpha_2\beta\alpha_1}).
\end{equation*}
This notation is very convenient for the compact representation of tensorial expressions with multiple overlapping symmetries.

With the help of the tensor algebra package \emph{xAct} for Mathematica~\cite{martin-garcia:aa,nutma:aa}, we obtained the following coincidence limits:
\begin{align*}
  [\sigma^{\mu}{}_{\alpha'\beta'}] &= 0 \\
  [\sigma^{\mu}{}_{\alpha'\alpha'\beta'}] &= \tfrac{2}{3} R^\mu{}_{\alpha \beta \alpha} \\
  [\sigma^{\mu}{}_{\alpha'\beta'\beta'}] &= -\tfrac{1}{3} R^\mu{}_{\beta \alpha \beta} \\
  [\sigma^{\mu}{}_{\alpha'\alpha'\alpha'\beta'}] &= \tfrac{1}{2} R^\mu{}_{\alpha \beta \alpha}{}_{;\alpha} \\
  [\sigma^{\mu}{}_{\alpha'\alpha'\beta'\beta'}] &= \tfrac{5}{6} R^{\mu }{}_{\alpha \beta \alpha}{}_{;\beta} - \tfrac{1}{6} R^\mu{}_{\beta \alpha \beta}{}_{;\alpha} \\
  [\sigma^{\mu}{}_{\alpha'\beta'\beta'\beta'}] &= -\tfrac{1}{2} R^\mu{}_{\beta \alpha \beta}{}_{;\beta} \\
  [\sigma^{\mu}{}_{\alpha'\alpha'\alpha'\alpha'\beta'}] &= \tfrac{2}{5} R^\mu{}_{\alpha \beta \alpha}{}_{;\alpha \alpha} + \tfrac{8}{15} R^\mu{}_{\alpha \gamma \alpha} R^\gamma{}_{\alpha \beta \alpha} \\
  \begin{split}
    [\sigma^{\mu}{}_{\alpha'\alpha'\alpha'\beta'\beta'}] &=
      \tfrac{7}{10} R^\mu{}_{\alpha \beta \alpha}{}_{;\beta \alpha}
    - \tfrac{1}{10} R^\mu{}_{\beta \alpha \beta}{}_{;\alpha \alpha} \\&\quad
    + \tfrac{7}{15} R^\mu{}_{\alpha \gamma \alpha} R^\gamma{}_{\beta \alpha \beta}
    + \tfrac{31}{15} R^\mu{}_{\alpha \gamma \beta} R^\gamma{}_{\alpha \beta \alpha}
    - \tfrac{8}{15} R^\mu{}_{\beta \gamma \alpha} R^\gamma{}_{\alpha \beta \alpha}
  \end{split} \\
  \begin{split}
    [\sigma^{\mu}{}_{\alpha'\alpha'\beta'\beta'\beta'}] &=
    - \tfrac{3}{10} R^\mu{}_{\beta \alpha \beta}{}_{;\alpha \beta}
    + \tfrac{9}{10} R^\mu{}_{\alpha \beta \alpha}{}_{;\beta \beta} \\&\quad
    + \tfrac{7}{15} R^\mu{}_{\beta \gamma \beta} R^\gamma{}_{\alpha \beta \alpha}
    + \tfrac{1}{15} R^\mu{}_{\beta \gamma \alpha} R^\gamma{}_{\beta \alpha \beta}
    - \tfrac{8}{15} R^\mu{}_{\alpha \gamma \beta} R^\gamma{}_{\beta \alpha \beta}
  \end{split} \\
  [\sigma^{\mu}{}_{\alpha'\beta'\beta'\beta'\beta'}] &= -\tfrac{3}{5} R^\mu{}_{\beta \alpha \beta}{}_{;\beta \beta} - \tfrac{7}{15} R^\mu{}_{\beta \gamma \beta} R^\gamma{}_{\beta \alpha \beta}.
\end{align*}
To simplify the resulting expressions, several identities such as the Bianchi identities were applied automatically via the \emph{xAct} Mathematica package.

Let us describe the methods we used to compute the coincidence limits above, and more generally, how one can compute the quantities
\begin{equation}\label{obta1}
  [\sigma_{\alpha_1\dotsm\alpha_n\beta_1'\dotsm\beta_m'}].
\end{equation}

As a first step we note that the following coincidence limits are immediate:
\begin{equation}\label{eq:synge-basic-coinc}
  [\sigma] = 0,
  \quad
  [\sigma_{\mu}] = 0,
  \quad
  [\sigma_{\mu\nu}] = [\sigma_{\mu'\nu'}] = g_{\mu\nu}.
\end{equation}
By raising indices, it follows from~\eqref{eq:synge-basic-coinc}
\begin{equation*}
  [\sigma^{\mu}{}_{\nu}] = \delta^{\mu}{}_{\nu}
\end{equation*}
and thus, applying Synge's rule~\eqref{eq:synge-rule},
\begin{equation*}
  [\sigma^{\mu}{}_{\nu'}] = -\delta^{\mu}{}_{\nu}.
\end{equation*}

\begin{proposition}
  For $n > 1$,
  \begin{equation*}
    [\sigma^{\mu}{}_{(\nu_1\dotsm\nu_n)}] = 0,
    \quad
    [\sigma^{\mu}{}_{(\nu_1'\dotsm\nu_n')}] = 0.
  \end{equation*}
\end{proposition}
\begin{proof}
  Since the proof of both identities proceeds analogously, we only prove the first identity.
  We apply $n$ symmetrized covariant derivatives to one of the identities of~\eqref{eq:synge-ident}
  \begin{equation*}
    \sigma^{\mu}{}_{\alpha} \sigma^{\alpha} - \sigma^{\mu}=0,
  \end{equation*}
  obtaining
  \begin{equation*}
    0 = \sigma^{\mu}{}_{\alpha} \sigma^{\alpha}{}_{(\nu_1\dotsm\nu_n)} + \sigma^{\mu}{}_{\alpha(\nu_1} \sigma^{\alpha}{}_{\nu_2\dotsm\nu_n)} + \dotsb + \sigma^{\mu}{}_{\alpha(\nu_1\dotsm\nu_n)} \sigma^{\alpha} - \sigma^{\mu}{}_{(\nu_1\dotsm\nu_n)}.
  \end{equation*}
  Then we take the coincidence limit and apply~\eqref{eq:synge-basic-coinc}:
  \begin{align}
    \begin{split}\nonumber
      0 &= \delta^\mu{}_\alpha [\sigma^{\alpha}{}_{(\nu_1\dotsm\nu_n)}] + [\sigma^{\mu}{}_{\alpha(\nu_1}] [\sigma^{\alpha}{}_{\nu_2\dotsm\nu_n)}] \\&\quad + \dotsb + [\sigma^{\mu}{}_{\alpha(\nu_1\dotsm\nu_{n-1}}] \delta^\alpha{}_{\nu_n} - [\sigma^{\mu}{}_{(\nu_1\dotsm\nu_n)}]
    \end{split}
    \\
    \begin{split}\label{eq:synge-n-deriv}
      &= [\sigma^{\mu}{}_{(\nu_1\dotsm\nu_n)}] + [\sigma^{\mu}{}_{\alpha(\nu_1}] [\sigma^{\alpha}{}_{\nu_2\dotsm\nu_n)}] \\&\quad + \dotsb + [\sigma^{\mu}{}_{\alpha(\nu_1\dotsm\nu_{n-2}}] [\sigma^{\alpha}{}_{\nu_{n-1}\nu_n)}].
    \end{split}
  \end{align}
  Specializing to $n=2$, this yields
  \begin{equation*}
    [\sigma^{\mu}{}_{(\nu_1\nu_2)}] = 0.
  \end{equation*}
  The result for $n>2$ is then obtained from~\eqref{eq:synge-n-deriv} by induction.
\end{proof}

Next note that to compute~\eqref{obta1} it is enough to do it first for all primed or all unprimed indices:
\begin{equation}\label{obta2}
  [\sigma_{\alpha_1\dotsm\alpha_n}] = [\sigma_{\alpha_1'\dotsm\alpha_n'}].
\end{equation}
Indeed, by Synge's rule~\eqref{eq:synge-rule},
\begin{equation*}
  [\sigma_{\alpha_1\dotsm\alpha_n\beta_1'\dotsm\beta_m'}]_{;\mu}
  = [\sigma_{\alpha_1\dotsm\alpha_n\mu\beta_1'\dotsm\beta_m'}] + [\sigma_{\alpha_1\dotsm\alpha_n\beta_1'\dotsm\beta_m'\mu'}],
\end{equation*}
so it is easy to add or remove primes.

Note that $\sigma_{\alpha_1\dotsm\alpha_n}$ is not symmetric with respect to permutation of indices except for the first two since covariant derivatives do not in general commute.
For instance, we have
\begin{equation}\label{obta6}
  \sigma_{\alpha\beta\gamma} = \sigma_{\alpha\gamma\beta} + \sigma_\mu R_{\gamma\beta\alpha}{}^\mu.
\end{equation}
For more complicated cases, we have the following lemma.

\begin{lemma}\label{lem:sigma_ordering}
  Let $\boldsymbol{\alpha}=(\alpha_1\dotsm\alpha_n)$ and $\boldsymbol{\nu}=(\nu_1\dotsm\nu_m)$.
  Then
  \begin{equation*}
    \sigma_{\boldsymbol{\alpha}\beta\gamma\boldsymbol{\nu}} =
    \sigma_{\boldsymbol{\alpha}\gamma\beta\boldsymbol{\nu}} + \dotsb,
  \end{equation*}
  where $\dotsb$ indicates terms with $\sigma$'s having at most $n+m$ indices.
\end{lemma}
\begin{proof}
  We have
  \begin{equation}\label{obta4}
    \sigma_{\boldsymbol{\alpha}\beta\gamma} = \sigma_{\boldsymbol{\alpha}\gamma\beta} + \sum_{j=1}^n \sigma_{\boldsymbol{\alpha}(i,\mu)} R_{\gamma\beta\alpha_i}{}^{\mu},
  \end{equation}
  where $\boldsymbol{\alpha}(i,\mu)$ is the multiindex coinciding with $\boldsymbol{\alpha}$ except that on the $i$th place there is $\mu$ instead of $\alpha_i$.
  This proves the lemma for $m=0$.

  Then we apply the covariant derivatives $\nabla_{\nu_m} \dotsm \nabla_{\nu_1}$ to both sides of~\eqref{obta4}.
  We obtain
  \begin{equation}\label{obta5}
    \sigma_{\boldsymbol{\alpha}\beta\gamma\boldsymbol{\nu}} = \sigma_{\boldsymbol{\alpha}\gamma\beta\boldsymbol{\nu}} + \sum_{j=1}^n \bigl(\sigma_{\boldsymbol{\alpha}(i,\mu)} R_{\gamma\beta\alpha_i}{}^\mu \bigr)_{;\boldsymbol{\nu}}.
  \end{equation}
  Clearly, after applying the Leibniz rule, the second term on the right of~\eqref{obta5} will not contain $\sigma$'s with more than $n+m$ indices.
\end{proof}

Let us now describe a recursive procedure to compute \eqref{obta2} relying on one of the identities~\eqref{eq:synge-ident}, viz.,
\begin{equation*}
  \sigma_{\alpha_1} = \sigma_{\beta\alpha_1} \sigma^{\beta}.
\end{equation*}
Applying $n-1$ covariant derivatives, this yields
\begin{equation*}
  \sigma_{\boldsymbol\alpha} = (\sigma_{\beta\alpha_1} \sigma^{\beta})_{;\alpha_2\dotsm\alpha_n}
\end{equation*}
with $\boldsymbol\alpha = (\alpha_1, \dotsc, \alpha_n)$.
Therefore, by the Leibniz rule and Lemma~\eqref{lem:sigma_ordering},
\begin{equation}\label{eq:sigma_recursion_step}
  \sigma_{\boldsymbol\alpha} = \sigma_{\beta\boldsymbol\alpha} \sigma^\beta + \sum_{i=1}^n \sigma_{\boldsymbol\alpha(i,\beta)} \sigma^\beta{}_{\alpha_i} + \dotsb,
\end{equation}
where, as above, $\boldsymbol{\alpha}(i,\beta)$ is the multiindex coinciding with $\boldsymbol{\alpha}$ except that on the $i$th place there is $\beta$ instead of $\alpha_i$, and $\dotsb$ indicates terms where no factor of $\sigma$ has more than $n-2$ indices.
Then we take the coincidence limit and use the basic coincidence limits~\eqref{eq:synge-basic-coinc} to obtain
\begin{equation*}
  [\sigma_{\boldsymbol\alpha}] = 0 + n [\sigma_{\boldsymbol\alpha}] + [\dotsb],
\end{equation*}
where $[\dotsb]$ indicates the coincidence limit of the $\dotsb$ terms in~\eqref{eq:sigma_recursion_step}.
Since $[\dotsb]$ contains no factor of $\sigma$ with more than $n-2$ indices, this gives the desired recursion.

\subsection{Covariant Taylor expansion}
\label{sec:cov-expansion}

Let $M \ni x \mapsto T(x)\in \T_x^{p,q}M$ be a tensor field.
Note also the following important fact:
\begin{equation*}
  \frac{\rd^n}{\rd \tau^n}\parap{T(x+\tau u)}_{\tau u}\Big|_{\tau=0}=(u\cdot\nabla)^nT(x).
\end{equation*}
Therefore, a (covariant) Taylor expansion for $T$ is given by
\begin{equation}\label{eq:taylor1}
  \parap{T(x+u)}_u \sim \exp(u \cdot \nabla) T(x) = \sum_n \frac{1}{n!} (u \cdot \nabla)^n T(x),
\end{equation}
where it is necessary to first parallel transport $T(x+u)$ to $x$.
Let us rewrite~\eqref{eq:taylor1} in coordinates:
\begin{equation*}
  g(x,x+u)^{\boldsymbol\gamma}{}_{\boldsymbol{\gamma'}} g(x,x+u)_{\boldsymbol{\delta}}{}^{\boldsymbol{\delta'}} T(x+u)_{\boldsymbol{\delta'}}^{\boldsymbol{\gamma'}} \sim \sum_{{\boldsymbol\rho}}\frac{1}{|\boldsymbol{\rho}|!} T(x)_{\boldsymbol{\delta};\boldsymbol{\rho}}^{\boldsymbol\gamma}u^{\boldsymbol{\rho}},
\end{equation*}
where $u^{\boldsymbol\beta} = u^{\beta_1} \dotsm u^{\beta_n}$.
We remark that formulas for remainder term are analogous to the usual Taylor expansion.

We will be especially interested in expansions of bitensors around the diagonal.
Let $(x,y)\mapsto T(x,y)$ be a bitensor.
Then~\eqref{eq:taylor1} can be rewritten as
\begin{equation*}
  \parap{T(x,x+u)}_u \sim [\exp(u \cdot \nabla') T](x) = \sum_n \frac{1}{n!} [(u \cdot \nabla')^n T](x),
\end{equation*}
where $\nabla'$ denotes covariant differentiation with respect to the second argument.
In coordinates, this can also be written as
\begin{equation*}
  g(x,x+u)^{\boldsymbol\gamma}{}_{\boldsymbol{\gamma'}} g(x,x+u)_{\boldsymbol{\delta}}{}^{\boldsymbol{\delta'}} T(x,x+u)_{\boldsymbol\alpha\boldsymbol\delta'}^{\boldsymbol\beta\boldsymbol\gamma'} \sim \sum_{\boldsymbol{\rho'}} \frac{1}{|\boldsymbol\rho|!} [T_{\boldsymbol\alpha\boldsymbol\delta';\boldsymbol\rho'}^{\boldsymbol\beta\boldsymbol{\gamma'}}](x) u^{\boldsymbol\rho}.
\end{equation*}

A particularly efficient approach to calculating coefficients of covariant expansions of many important bitensors is Avramidi's method~\cite{avramidi:aa,avramidi:ab,decanini:aa}, especially the semi-recursive variant presented in~\cite{ottewill:ab}.
Avramidi's method relies on deriving recursion relations for the coefficients from certain transport equations (such as~\eqref{eq:synge-ident}) for the bitensor.
We refer to~\cite{ottewill:ab} for a full explanation and several examples.

The Taylor expansion for tensor fields~\eqref{eq:taylor1} generalizes to tensor-valued functions on the cotangent bundle via the horizontal derivative.
\begin{equation}\label{eq:taylor2}
  \parap[\big]{S(x+u,\parap{p}^u)}_u \sim \exp(u \cdot \nnabla) S(x,p) = \sum_n \frac{1}{n!} (u \cdot \nnabla)^n S(x,p).
\end{equation}

\subsection{Geodesic triangle}
\label{sec:triangle}

Consider three points $x,y,z$ in a geodesically convex neighbourhood of the diagonal.
By connecting these three points by the distinguished geodesics between them, we obtain a geodesic triangle.

We define the following vectors:
\begin{align*}
  v &\defn (y-x) \in \T_xM, \\
  w &\defn (z-x) \in \T_xM, \\
  u' &\defn (z-y) \in \T_yM, \\
  u &\defn \parap{u'}_v \in \T_xM.
\end{align*}
The arrangement of these vectors is schematically depicted in Fig.~\ref{fig:geodesic-triangle}.

\begin{figure}
  \centering
  \begin{tikzpicture}
    \fill (3.5,3) circle [radius=2pt] node[above] {$z$};
    \fill (0,0) circle [radius=2pt] node[below left=-1pt] {$x$};
    \fill (6,0) circle [radius=2pt] node[below right=-1pt] {$y$};

    \draw[->,shorten >=3pt,thick] (0,0) -- node[below] {$v$} (6,0);
    \draw[->,shorten >=3pt,thick] (6,0) -- node[above right=-2pt] {$u'$} (3.5,3);
    \draw[->,shorten >=3pt,thick] (0,0) -- node[above left=-2pt] {$w$} (3.5,3);
    \draw[->,shorten >=3pt,thick] (0,0) -- node[above right=-2pt] {$u$} (-2.7,3.1);
  \end{tikzpicture}
  \caption{\label{fig:geodesic-triangle} A geodesic triangle spanned by three points $x,y,z$. $v$ is the tangent at~$x$ to the geodesic from~$x$ to~$y$; $w$ is the tangent at~$x$ to the geodesic from~$x$ to~$z$; $u'$ is the tangent at~$y$ to the geodesic from~$y$ to~$z$; $u$ is the parallel transport of~$u'$ from~$y$ to~$x$ along the geodesic given by~$v$.}
\end{figure}
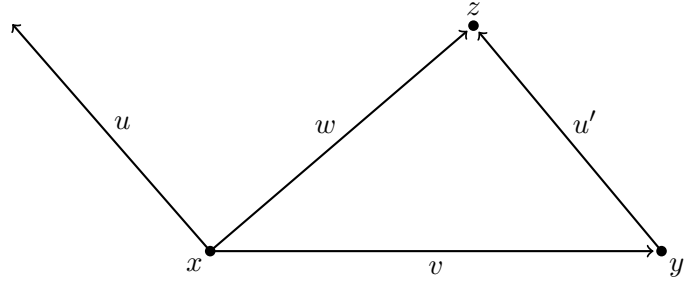

While $w = u + v$ on flat spaces, due to the effects of curvature on parallel transport this is no longer true on generic curved spaces, i.e., the triangle formed by these vectors does not close.

We will consider two cases:
First, suppose that we are given $v,w$, and $u$ is unknown.
In other words,
\begin{equation*}
  u = \parap{(x+w)-(x+v)}_v,
\end{equation*}
or, in terms of the world function with $y=x+v$, $z=x+w$,
\begin{equation*}
  u^\mu = -g(x,y)^\mu{}_{\alpha'} \sigma(y,z)^{\alpha}.
\end{equation*}
We perform two covariant expansions (with base point~$x$) to find
\begin{align*}
  u^\mu
  &\sim -g(x,y)^\mu{}_{\alpha'} \sum_n \frac{1}{n!} (w \cdot \nabla')^n \sigma(y,x)^{\alpha} \\
  &\sim -\sum_{m,n} \frac{1}{m!n!} [(v \cdot \nabla)^m (w \cdot \nabla')^n \sigma^{\mu}](x) \\
  &= -\sum_{\boldsymbol\alpha,\boldsymbol\beta} \frac{1}{|\boldsymbol\alpha|!|\boldsymbol\beta|!} v^{\boldsymbol\alpha} w^{\boldsymbol\beta} [\sigma^{\mu}{}_{(\boldsymbol\alpha)(\boldsymbol\beta')}](x),
\end{align*}
where we used~\eqref{eq:covd_parallel_transport}.

Secondly, suppose that we are given $u,v$, and $w$ is unknown.
In other words,
\begin{equation*}
  w = \Bigl(\bigl((x+v)+\parap{u}^v\bigr)-x\Bigr),
\end{equation*}
or in terms of the world function with $z=(x+v)+\parap{u}^v$,
\begin{equation*}
  w^\mu = -\sigma(x, z)^{\mu}.
\end{equation*}
We perform two covariant expansions (first with base point $y=x+v$ and then with base point~$x$) to find
\begin{align}
  w^\mu
  &\sim -\sum_n \frac{1}{n!} (u' \cdot \nabla')^n \sigma(x, y)^{\mu} \nonumber\\
  &\sim -\sum_{m,n} \frac{1}{m!n!} [(v \cdot \nabla')^m (u \cdot \nabla')^n \sigma^{\mu}](x) \nonumber\\
  &= -\sum_{\boldsymbol\alpha,\boldsymbol\beta}\frac{1}{|\boldsymbol\alpha|!|\boldsymbol\beta|!} v^{\boldsymbol\alpha} u^{\boldsymbol\beta}[\sigma^{\mu}{}_{(\boldsymbol\alpha')(\boldsymbol\beta')}](x), \label{eq:qe}
\end{align}
where we used~\eqref{eq:covd_parallel_transport}.

Set
\begin{equation*}
  w = (u + v) + \delta(u,v),
\end{equation*}
where $\delta(u,v)$ specifies the \emph{geodesic defect}, i.e., the failure of $u,v,w$ to form a triangle due to the effects of curvature.
In other words,
\begin{equation*}
  \delta(u,v) = \Bigl(\bigl((x+v)+\parap{u}^v\bigr)-x\Bigr)-u-v.
\end{equation*}
Since
\begin{equation*}
  [\sigma^{\mu}{}_{(\boldsymbol\nu')}] = \begin{cases}
    -\delta^\mu{}_\nu & \quad\text{for}\;\abs{\boldsymbol\nu} = 1, \\
    0 & \quad\text{otherwise},
  \end{cases}
\end{equation*}
we find from \eqref{eq:qe}
\begin{subequations}\label{eq:geodesic_defect}\begin{align}
  \delta(u,v)^\mu
  &= -\sum_{m,n \geq 1} \frac{1}{m!n!} [(v \cdot \nabla')^m (u \cdot \nabla')^n \sigma^{\mu}] \\
  &= -\sum_{\substack{\boldsymbol\alpha, \boldsymbol\beta \\ \mathclap{\abs{\boldsymbol\alpha}, \abs{\boldsymbol\beta} \geq 1}}} \frac{1}{\abs{\boldsymbol\alpha}!\abs{\boldsymbol\beta}!} u^{\boldsymbol\alpha} v^{\boldsymbol\beta} [\sigma^{\mu}{}_{(\boldsymbol\alpha')(\boldsymbol\beta')}].
\end{align}\end{subequations}

\subsection{Expansion of the Van Vleck--Morette determinant}
\label{sec:cov-exp-vanvleck}

Following~\cite{avramidi:aa,avramidi:ab}, we define the symmetric biscalar
\begin{equation*}
  \zeta(x,y) \defn \log\Delta(x,y)^\frac12.
\end{equation*}
It satisfies the transport equation
\begin{equation}\label{eq:zeta-transport}
  \cD' \zeta = \frac12 (4 - \sigma^{\alpha\rlapprime}{}_{\alpha'}).
\end{equation}

For the remainder of this section, let $u,v \in \T_zM$.
Using Avramidi's method~\cite{avramidi:aa,avramidi:ab,ottewill:ab,decanini:aa} applied to the transport equation~\eqref{eq:zeta-transport}, it is easy to calculate the covariant expansion of $\zeta(z,z+u)$ to high orders.
Up to fourth order, it is given by
\begin{equation}\begin{split}\label{eq:zeta_expansion-1}
  \zeta(z,z+u) &=
  \tfrac{1}{12} R_{\alpha_1 \alpha_2} u^{\alpha_1} u^{\alpha_2}
  + \tfrac{1}{24} R_{\alpha_1 \alpha_2 ; \alpha_3} u^{\alpha_1} u^{\alpha_2} u^{\alpha_3} \\&\quad
  + \bigl( \tfrac{1}{360} R^{\beta}{}_{\alpha_1 \gamma \alpha_2} R^{\gamma}{}_{\alpha_3 \beta \alpha_4} + \tfrac{1}{80} R_{\alpha_1 \alpha_2 ; \alpha_3 \alpha_4} \bigr) u^{\alpha_1} u^{\alpha_2} u^{\alpha_3} u^{\alpha_4}
  + \dotsb
\end{split}\end{equation}
where all tensors here and below are evaluated at $z$ unless otherwise indicated.
For an expansion up to 11th order we refer to~\cite{decanini:aa}.
This calculation can be automated with the \emph{CovariantSeries} package for Mathematica~\cite{ottewill:ab,wardell:aa}.

Applying a (covariant) Taylor expansion to the expansion above, we obtain
\begin{equation*}\begin{split}\MoveEqLeft
  \zeta(z+v,z+v+\parap{u}^v) \\&=
  \tfrac{1}{12} R_{\alpha_1 \alpha_2} u^{\alpha_1} u^{\alpha_2}
  + \tfrac{1}{24} R_{\alpha_1 \alpha_2 ; \alpha_3} u^{\alpha_1} u^{\alpha_2} u^{\alpha_3} + \tfrac{1}{12} R_{\alpha_1 \alpha_2 ; \beta_1} u^{\alpha_1} u^{\alpha_2} v^{\beta_1} \\&\quad
  + \bigl( \tfrac{1}{360} R^{\beta}{}_{\alpha_1 \gamma \alpha_2} R^{\gamma}{}_{\alpha_3 \beta \alpha_4} + \tfrac{1}{80} R_{\alpha_1 \alpha_2 ; \alpha_3 \alpha_4} \bigr) u^{\alpha_1} u^{\alpha_2} u^{\alpha_3} u^{\alpha_4} \\&\quad + \tfrac{1}{24} R_{\alpha_1 \alpha_2 ; \alpha_3 \beta_1} u^{\alpha_1} u^{\alpha_2} u^{\alpha_3} v^{\beta_1} + \tfrac{1}{12} R_{\alpha_1 \alpha_2 ; \beta_1 \beta_2} u^{\alpha_1} u^{\alpha_2} v^{\beta_1} v^{\beta_2}
  + \dotsb
\end{split}\end{equation*}

With the help of Synge's rule, the coefficients of the covariant expansion of $\zeta(z-u,z+u)$ can be obtained from those of $\zeta(z,z+u)$.
A helpful fact for this calculation is the symmetry of $\zeta$.
Among other things, it implies that only even order coefficients are non-vanishing.
We obtain, up to fourth order,
\begin{equation}\begin{split}\label{eq:zeta_expansion-3}
  \zeta(z-u,z+u) &=
  \tfrac{1}{3} R_{\alpha_1 \alpha_2} u^{\alpha_1} u^{\alpha_2} \\&\quad
  + \bigl( \tfrac{2}{45} R^{\beta}{}_{\alpha_1 \gamma \alpha_2} R^{\gamma}{}_{\alpha_3 \beta \alpha_4} + \tfrac{1}{30} R_{\alpha_1 \alpha_2 ; \alpha_3 \alpha_4} \bigr) u^{\alpha_1} u^{\alpha_2} u^{\alpha_3} u^{\alpha_4}
  + \dotsb
\end{split}\end{equation}

Combining the latter expansion with an additional (covariant) Taylor expansion, we find (again up to fourth order)
\begin{equation}\begin{split}\label{eq:zeta_expansion-4}\MoveEqLeft
  \zeta(z+v-\parap{u}^v,z+v+\parap{u}^v) \\&=
  \tfrac13 R_{\alpha_1 \alpha_2} u^{\alpha_1} u^{\alpha_2}
  + \tfrac13 R_{\alpha_1 \alpha_2 ; \beta_1} u^{\alpha_1} u^{\alpha_2} v^{\beta_1} \\&\quad
  + \bigl( \tfrac{2}{45} R^{\beta}{}_{\alpha_1 \gamma \alpha_2} R^{\gamma}{}_{\alpha_3 \beta \alpha_4} + \tfrac{1}{30} R_{\alpha_1 \alpha_2 ; \alpha_3 \alpha_4} \bigr) u^{\alpha_1} u^{\alpha_2} u^{\alpha_3} u^{\alpha_4} \\&\quad + \tfrac16 R_{\alpha_1 \alpha_2 ; \beta_1 \beta_2} u^{\alpha_1} u^{\alpha_2} v^{\beta_1} v^{\beta_2}
  + \dotsb
\end{split}\end{equation}

\subsection{Four geodesic triangles}

The vectors $u_1, u_2, v_1, v_2, w, \tilde{w}$ in the proof of Thm.~\ref{thm:star-product} form several geodesic triangles, in particular, the four triangles shown in different colors in Fig.~\ref{fig:star-product-triangle}.
In formulas, we have
\begin{subequations}\label{eq:4triangles}\begin{align}
  w &= -u_1 + v_2 + \delta(-u_1,v_2), \\
  -w &= -u_2 + v_1 + \delta(-u_2,v_1), \\
  \tilde{w} &= u_2 + v_1 + \delta(u_2,v_1), \\
  \tilde{w} &= u_1 + v_2 + \delta(u_1,v_2),
\end{align}\end{subequations}
where we use the geodesic defect $\delta(\cdot\,,\cdot)$ defined in~\eqref{eq:geodesic_defect}.

Define the even and odd parts of the geodesic defect as
\begin{equation*}
  \delta_\pm(u,v) = \frac12 \bigl( \delta(u,v) \pm \delta(-u,v) \bigr),
  \quad
  u,v \in \T_zM.
\end{equation*}
Then, solving~\eqref{eq:4triangles} for $v_1, v_2, w, \tilde{w}$, we find
\begin{align*}
  v_1 &= u_1 + \delta_-(u_1,v_2) - \delta_+(u_2,v_1), \\
  v_2 &= u_2 - \delta_+(u_1,v_2) + \delta_-(u_2,v_1), \\
  w &= -u_1 + u_2 - \delta_-(u_1,v_2) + \delta_-(u_2,v_1), \\
  \tilde{w} &= u_1 + u_2 + \delta_-(u_1,v_2) + \delta_-(u_2,v_1).
\end{align*}
These equations can be solved perturbatively to obtain an expressions for $v_1,v_2,w$ and~$\tilde{w}$ depending only on~$u_1$ and~$u_2$.
Up to 5th order,
\begin{align}
  \begin{split}\label{eq:v1-expansion}
    v_1^\mu &=
    u_1^\mu
    + \tfrac12 R^\mu{}_{\alpha_1\alpha_2\alpha_3} u_2^{\alpha_1} u_1^{\alpha_2} u_2^{\alpha_3} \\
    &\quad + R^\mu{}_{\alpha_1\alpha_2\alpha_3;\alpha_4} \bigl(-\tfrac{1}{24} u_1^{\alpha_1} u_2^{\alpha_2} u_1^{\alpha_3} u_2^{\alpha_4} + \tfrac{5}{24} u_2^{\alpha_1} u_1^{\alpha_2} u_2^{\alpha_3} u_1^{\alpha_4} \\&\qquad -\tfrac{1}{12} u_1^{\alpha_1} u_2^{\alpha_2} u_1^{\alpha_3} u_1^{\alpha_4} + \tfrac{1}{12} u_2^{\alpha_1} u_1^{\alpha_2} u_2^{\alpha_3} u_2^{\alpha_4} \bigr) \\
    &\quad + R^\mu{}_{\alpha_1\alpha_2\alpha_3;\alpha_4\alpha_5} \bigl( \tfrac{1}{24} u_2^{\alpha_1} u_1^{\alpha_2} u_2^{\alpha_3} u_2^{\alpha_4} u_2^{\alpha_5} - \tfrac{1}{12} u_1^{\alpha_1} u_2^{\alpha_2} u_1^{\alpha_3} u_1^{\alpha_4} u_2^{\alpha_5} \\&\qquad + \tfrac{1}{12} u_2^{\alpha_1} u_1^{\alpha_2} u_2^{\alpha_3} u_1^{\alpha_4} u_1^{\alpha_5} \bigr) \\&\quad + R^\mu{}_{\alpha_1\beta\alpha_2} R^\beta{}_{\alpha_3\alpha_4\alpha_5} \bigl( \tfrac{5}{24} u_2^{\alpha_1} u_2^{\alpha_2} u_2^{\alpha_3} u_1^{\alpha_4} u_2^{\alpha_5} - \tfrac{1}{6} u_2^{\alpha_1} u_1^{\alpha_2} u_1^{\alpha_3} u_2^{\alpha_4} u_1^{\alpha_5} \\&\qquad + \tfrac{1}{12} u_1^{\alpha_1} u_1^{\alpha_2} u_2^{\alpha_3} u_1^{\alpha_4} u_2^{\alpha_5} \bigr)
    + \dotsb
  \end{split}
  \\
  \begin{split}\label{eq:v2-expansion}
    v_2^\mu &=
    u_2^\mu
    + \tfrac12 R^\mu{}_{\alpha_1\alpha_2\alpha_3} u_1^{\alpha_1} u_2^{\alpha_2} u_1^{\alpha_3} \\
    &\quad + R^\mu{}_{\alpha_1\alpha_2\alpha_3;\alpha_4} \bigl(-\tfrac{1}{24} u_2^{\alpha_1} u_1^{\alpha_2} u_2^{\alpha_3} u_1^{\alpha_4} + \tfrac{5}{24} u_1^{\alpha_1} u_2^{\alpha_2} u_1^{\alpha_3} u_2^{\alpha_4} \\&\qquad -\tfrac{1}{12} u_2^{\alpha_1} u_1^{\alpha_2} u_2^{\alpha_3} u_2^{\alpha_4} + \tfrac{1}{12} u_1^{\alpha_1} u_2^{\alpha_2} u_1^{\alpha_3} u_1^{\alpha_4} \bigr) \\
    &\quad + R^\mu{}_{\alpha_1\alpha_2\alpha_3;\alpha_4\alpha_5} \bigl(\tfrac{1}{24} u_1^{\alpha_1} u_2^{\alpha_2} u_1^{\alpha_3} u_1^{\alpha_4} u_1^{\alpha_5} - \tfrac{1}{12} u_2^{\alpha_1} u_1^{\alpha_2} u_2^{\alpha_3} u_2^{\alpha_4} u_1^{\alpha_5} \\&\qquad + \tfrac{1}{12} u_1^{\alpha_1} u_2^{\alpha_2} u_1^{\alpha_3} u_2^{\alpha_4} u_2^{\alpha_5} \bigr) \\&\quad + R^\mu{}_{\alpha_1\beta\alpha_2} R^\beta{}_{\alpha_3\alpha_4\alpha_5} \bigl(\tfrac{5}{24} u_1^{\alpha_1} u_1^{\alpha_2} u_1^{\alpha_3} u_2^{\alpha_4} u_1^{\alpha_5} - \tfrac{1}{6} u_1^{\alpha_1} u_2^{\alpha_2} u_2^{\alpha_3} u_1^{\alpha_4} u_2^{\alpha_5} \\&\qquad + \tfrac{1}{12} u_2^{\alpha_1} u_2^{\alpha_2} u_1^{\alpha_3} u_2^{\alpha_4} u_1^{\alpha_5} \bigr)
    + \dotsb
  \end{split}
  \\
  \begin{split}\label{eq:w-expansion}
    w^\mu &=
    -u_1^\mu + u_2^\mu
    + R^\mu{}_{\alpha_1\alpha_2\alpha_3} \bigl( \tfrac16 u_1^{\alpha_1} u_2^{\alpha_2} u_1^{\alpha_3} - \tfrac16 u_2^{\alpha_1} u_1^{\alpha_2} u_2^{\alpha_3} \bigr) \\
    &\quad + R^\mu{}_{\alpha_1\alpha_2\alpha_3;\alpha_4} \bigl( -\tfrac{1}{6} u_2^{\alpha_1} u_1^{\alpha_2} u_2^{\alpha_3} u_2^{\alpha_4} + \tfrac{1}{6} u_1^{\alpha_1} u_2^{\alpha_2} u_1^{\alpha_3} u_1^{\alpha_4} \bigr) \\
    &\quad + R^\mu{}_{\alpha_1\alpha_2\alpha_3;\alpha_4\alpha_5} \bigl(-\tfrac{7}{120} u_2^{\alpha_1} u_1^{\alpha_2} u_2^{\alpha_3} u_2^{\alpha_4} u_1^{\alpha_5} + \tfrac{1}{120} u_1^{\alpha_1} u_2^{\alpha_2} u_1^{\alpha_3} u_2^{\alpha_4} u_2^{\alpha_5} \\&\qquad + \tfrac{1}{40} u_1^{\alpha_1} u_2^{\alpha_2} u_1^{\alpha_3} u_1^{\alpha_4} u_1^{\alpha_5} + \tfrac{7}{120} u_1^{\alpha_1} u_2^{\alpha_2} u_1^{\alpha_3} u_1^{\alpha_4} u_2^{\alpha_5} \\&\qquad - \tfrac{1}{120} u_2^{\alpha_1} u_1^{\alpha_2} u_2^{\alpha_3} u_1^{\alpha_4} u_1^{\alpha_5} - \tfrac{1}{40} u_2^{\alpha_1} u_1^{\alpha_2} u_2^{\alpha_3} u_2^{\alpha_4} u_2^{\alpha_5} \bigr) \\
    &\quad + R^\mu{}_{\alpha_1\beta\alpha_2} R^\beta{}_{\alpha_3\alpha_4\alpha_5} \bigl(\tfrac{7}{360} u_2^{\alpha_1} u_2^{\alpha_2} u_1^{\alpha_3} u_2^{\alpha_4} u_1^{\alpha_5} - \tfrac{11}{360} u_2^{\alpha_1} u_1^{\alpha_2} u_2^{\alpha_3} u_1^{\alpha_4} u_2^{\alpha_5} \\&\qquad - \tfrac{11}{90} u_1^{\alpha_1} u_2^{\alpha_2} u_2^{\alpha_3} u_1^{\alpha_4} u_2^{\alpha_5} + \tfrac{7}{360} u_1^{\alpha_1} u_1^{\alpha_2} u_1^{\alpha_3} u_2^{\alpha_4} u_1^{\alpha_5} \\&\qquad - \tfrac{7}{360} u_1^{\alpha_1} u_1^{\alpha_2} u_2^{\alpha_3} u_1^{\alpha_4} u_2^{\alpha_5} + \tfrac{11}{360} u_1^{\alpha_1} u_2^{\alpha_2} u_1^{\alpha_3} u_2^{\alpha_4} u_1^{\alpha_5} \\&\qquad + \tfrac{11}{90} u_2^{\alpha_1} u_1^{\alpha_2} u_1^{\alpha_3} u_2^{\alpha_4} u_1^{\alpha_5} - \tfrac{7}{360} u_2^{\alpha_1} u_2^{\alpha_2} u_2^{\alpha_3} u_1^{\alpha_4} u_2^{\alpha_5} \bigr)
    + \dotsb
  \end{split}
  \\
  \begin{split}\label{eq:tilde-w-expansion}
    \tilde{w}^\mu &=
    u_1^\mu + u_2^\mu
    + R^\mu_{\alpha_1\alpha_2\alpha_3} \bigl(\tfrac16 u_1^{\alpha_1} u_2^{\alpha_2} u_1^{\alpha_3} + \tfrac16 u_2^{\alpha_1} u_1^{\alpha_2} u_2^{\alpha_3}\bigr) \\
    %
    &\quad + R^\mu{}_{\alpha_1\alpha_2\alpha_3;\alpha_4;\alpha_5} \bigl(-\tfrac{7}{120} u_2^{\alpha_1} u_1^{\alpha_2} u_2^{\alpha_3} u_2^{\alpha_4} u_1^{\alpha_5} + \tfrac{1}{120} u_1^{\alpha_1} u_2^{\alpha_2} u_1^{\alpha_3} u_2^{\alpha_4} u_2^{\alpha_5} \\&\qquad + \tfrac{1}{40} u_1^{\alpha_1} u_2^{\alpha_2} u_1^{\alpha_3} u_1^{\alpha_4} u_1^{\alpha_5} -\tfrac{7}{120} u_1^{\alpha_1} u_2^{\alpha_2} u_1^{\alpha_3} u_1^{\alpha_4} u_2^{\alpha_5} \\&\qquad + \tfrac{1}{120} u_2^{\alpha_1} u_1^{\alpha_2} u_2^{\alpha_3} u_1^{\alpha_4} u_1^{\alpha_5} + \tfrac{1}{40} u_2^{\alpha_1} u_1^{\alpha_2} u_2^{\alpha_3} u_2^{\alpha_4} u_2^{\alpha_5} \bigr) \\
    &\quad + R^\mu{}_{\alpha_1\beta\alpha_2} R^\beta{}_{\alpha_3\alpha_4\alpha_5} \bigl(\tfrac{7}{360} u_2^{\alpha_1} u_2^{\alpha_2} u_1^{\alpha_3} u_2^{\alpha_4} u_1^{\alpha_5} - \tfrac{11}{360} u_2^{\alpha_1} u_1^{\alpha_2} u_2^{\alpha_3} u_1^{\alpha_4} u_2^{\alpha_5} + \\&\qquad - \tfrac{11}{90} u_1^{\alpha_1} u_2^{\alpha_2} u_2^{\alpha_3} u_1^{\alpha_4} u_2^{\alpha_5} + \tfrac{7}{360} u_1^{\alpha_1} u_1^{\alpha_2} u_1^{\alpha_3} u_2^{\alpha_4} u_1^{\alpha_5} \\&\qquad + \tfrac{7}{360} u_1^{\alpha_1} u_1^{\alpha_2} u_2^{\alpha_3} u_1^{\alpha_4} u_2^{\alpha_5} - \tfrac{11}{360} u_1^{\alpha_1} u_2^{\alpha_2} u_1^{\alpha_3} u_2^{\alpha_4} u_1^{\alpha_5} \\&\qquad - \tfrac{11}{90} u_2^{\alpha_1} u_1^{\alpha_2} u_1^{\alpha_3} u_2^{\alpha_4} u_1^{\alpha_5} + \tfrac{7}{360} u_2^{\alpha_1} u_2^{\alpha_2} u_2^{\alpha_3} u_1^{\alpha_4} u_2^{\alpha_5} \bigr)
    + \dotsb
  \end{split}
\end{align}
Note that the expansions for $v_1$ and $v_2$ can be obtained from one another by exchanging $u_1$ and $u_2$.
Moreover,
\begin{align*}
  \begin{split}
    \delta_+(u_1, v_2)^\mu &=
    -\tfrac{1}{3} R^\mu{}_{\alpha_1 \alpha_2 \alpha_3} u_1^{\alpha_1} u_1^{\alpha_3} u_2^{\alpha_2} \\&\quad
    + R^\mu{}_{\alpha_1 \alpha_2 \alpha_3 ; \alpha_4} \bigl( \tfrac{1}{24} u_2^{\alpha_1} u_1^{\alpha_2} u_2^{\alpha_3} u_1^{\alpha_4} - \tfrac{5}{24} u_1^{\alpha_1} u_2^{\alpha_2} u_1^{\alpha_3} u_2^{\alpha_4} \bigr) \\&\quad
    + R^\mu{}_{\alpha_1 \alpha_2 \alpha_3 ; \alpha_4 \alpha_5} \bigl(- \tfrac{1}{60} u_1^{\alpha_1} u_2^{\alpha_2} u_1^{\alpha_3} u_1^{\alpha_4} u_1^{\alpha_5} - \tfrac{3}{40} u_1^{\alpha_1} u_2^{\alpha_2} u_1^{\alpha_3} u_2^{\alpha_4} u_2^{\alpha_5} \\&\qquad + \tfrac{1}{40} u_2^{\alpha_1} u_1^{\alpha_2} u_2^{\alpha_3} u_1^{\alpha_4} u_2^{\alpha_5} \bigr) \\&\quad
    + R^\mu{}_{\alpha_1 \nu \alpha_2} R^\nu{}_{\alpha_3 \alpha_4 \alpha_5} \bigl( \tfrac{1}{6} u_1^{\alpha_1} u_1^{\alpha_2} u_1^{\alpha_3} u_1^{\alpha_4} u_2^{\alpha_5} - \tfrac{1}{45} u_1^{\alpha_1} u_1^{\alpha_2} u_1^{\alpha_3} u_2^{\alpha_4} u_1^{\alpha_5} \\&\qquad
    + \tfrac{2}{45} u_1^{\alpha_1} u_2^{\alpha_2} u_2^{\alpha_3} u_1^{\alpha_4} u_2^{\alpha_5} - \tfrac{1}{180} u_2^{\alpha_1} u_1^{\alpha_2} u_2^{\alpha_3} u_1^{\alpha_4} u_2^{\alpha_5} \\&\qquad - \tfrac{7}{180} u_2^{\alpha_1} u_2^{\alpha_2} u_1^{\alpha_3} u_2^{\alpha_4} u_1^{\alpha_5} \bigr)
    + \dotsb
  \end{split}
  \\
  \begin{split}
    \delta_-(u_1, v_2) &=
    \tfrac{1}{6} R^\mu{}_{\alpha_1 \alpha_2 \alpha_3} u_2^{\alpha_1} u_1^{\alpha_2} u_2^{\alpha_3} \\&\quad
    + R^\mu{}_{\alpha_1 \alpha_2 \alpha_3 ; \alpha_4} \bigl(- \tfrac{1}{12} u_1^{\alpha_1} u_2^{\alpha_2} u_1^{\alpha_3} u_1^{\alpha_4} + \tfrac{1}{12} u_2^{\alpha_1} u_1^{\alpha_2} u_2^{\alpha_3} u_2^{\alpha_4} \bigr) \\&\quad
    + R^\mu{}_{\alpha_1 \alpha_2 \alpha_3 ; \alpha_4 \alpha_5} \bigl(- \tfrac{7}{120} u_1^{\alpha_1} u_2^{\alpha_2} u_1^{\alpha_3} u_1^{\alpha_4} u_2^{\alpha_5} + \tfrac{1}{120} u_2^{\alpha_1} u_1^{\alpha_2} u_2^{\alpha_3} u_1^{\alpha_4} u_1^{\alpha_5} \\&\qquad + \tfrac{1}{40} u_2^{\alpha_1} u_1^{\alpha_2} u_2^{\alpha_3} u_2^{\alpha_4} u_2^{\alpha_5} \bigr) \\&\quad
    + R^\mu{}_{\alpha_1 \nu \alpha_2} R^\nu{}_{\alpha_3 \alpha_4 \alpha_5} \bigl( \tfrac{7}{360} u_1^{\alpha_1} u_1^{\alpha_2} u_2^{\alpha_3} u_1^{\alpha_4} u_2^{\alpha_5} - \tfrac{41}{360} u_1^{\alpha_1} u_2^{\alpha_2} u_1^{\alpha_3} u_2^{\alpha_4} u_1^{\alpha_5} \\&\qquad - \tfrac{1}{12} u_1^{\alpha_1} u_2^{\alpha_2} u_1^{\alpha_3} u_1^{\alpha_4} u_2^{\alpha_5} + \tfrac{1}{6} u_2^{\alpha_1} u_1^{\alpha_2} u_1^{\alpha_3} u_1^{\alpha_4} u_2^{\alpha_5} \\&\qquad + \tfrac{2}{45} u_2^{\alpha_1} u_1^{\alpha_2} u_1^{\alpha_3} u_2^{\alpha_4} u_1^{\alpha_5} + \tfrac{7}{360} u_2^{\alpha_1} u_2^{\alpha_2} u_2^{\alpha_3} u_1^{\alpha_4} u_2^{\alpha_5} \bigr)
    + \dotsb
  \end{split}
\end{align*}
and $\delta_\pm(u_2, v_1)$ are obtained by exchanging $u_1$ and $u_2$.

\subsection{Expansion of the geometric factor}
\label{sec:jacobian}

To calculate the Jacobian determinant in the proof of Thm.~\ref{thm:star-product}, we use
\begin{align*}
  -\tfrac12(w - \tilde{w}) &= u_1 + \delta_-(u_1, v_2), \\
  \tfrac12(w + \tilde{w}) &= u_2 + \delta_-(u_2, v_1).
\end{align*}
Therefore,
\begin{equation*}
  \abs[\bigg]{\frac{\partial(w,\tilde{w})}{\partial(u_1,u_2)}} = 2^d \abs[\bigg]{\one + \frac{\partial\bigl(\delta_-(u_1, v_2),\delta_-(u_2, v_1)\bigr)}{\partial(u_1,u_2)}},
\end{equation*}
with $v_1$ and $v_2$ understood as functions of $u_1, u_2$.

Using the series expansion of the logarithm, it holds
\begin{equation*}
  \det(\one + A) \sim \exp\Bigl(-\sum_k \frac{(-1)^k}{k} \tr(A^k)\Bigr)
\end{equation*}
for a square matrix~$A$, and thus we find (up to fourth order)
\begin{align*}
  \smash{\abs[\bigg]{\frac{\partial(w,\tilde{w})}{\partial(u_1,u_2)}}} = 2^d \exp\bigl[&
  R_{\alpha_1 \alpha_2} \bigl( \tfrac16 u_1^{\alpha_1} u_1^{\alpha_2} + \tfrac16 u_2^{\alpha_1} u_2^{\alpha_2} \bigr) \\&
  + R_{\alpha_1 \alpha_2 ; \alpha_3} \bigl( \tfrac{1}{12} u_1^{\alpha_1} u_1^{\alpha_2} u_1^{\alpha_3} - \tfrac{1}{12} u_1^{\alpha_1} u_1^{\alpha_2} u_2^{\alpha_3} \\&\quad + \tfrac{1}{6} u_1^{\alpha_1} u_2^{\alpha_2} u_1^{\alpha_3} + \tfrac{1}{6} u_2^{\alpha_1} u_1^{\alpha_2} u_2^{\alpha_3} \\&\quad - \tfrac{1}{12} u_2^{\alpha_1} u_2^{\alpha_2} u_1^{\alpha_3} + \tfrac{1}{12} u_2^{\alpha_1} u_2^{\alpha_2} u_2^{\alpha_3}) \\&
  + R_{\alpha_1 \alpha_2 ; \alpha_3 \alpha_4} \bigl( \tfrac{1}{40} u_1^{\alpha_1} u_1^{\alpha_2} u_1^{\alpha_3} u_1^{\alpha_4} - \tfrac{1}{30} u_1^{\alpha_1} u_1^{\alpha_2} u_2^{\alpha_3} u_2^{\alpha_4} \\&\quad + \tfrac{1}{10} u_1^{\alpha_1} u_2^{\alpha_2} u_1^{\alpha_3} u_2^{\alpha_4} + \tfrac{1}{10} u_2^{\alpha_1} u_1^{\alpha_2} u_2^{\alpha_3} u_1^{\alpha_4} \\&\quad - \tfrac{1}{30} u_2^{\alpha_1} u_2^{\alpha_2} u_1^{\alpha_3} u_1^{\alpha_4} + \tfrac{1}{40} u_2^{\alpha_1} u_2^{\alpha_2} u_2^{\alpha_3} u_2^{\alpha_4} \bigr) \\&
  + R_{\mu \alpha_1} R^\mu{}_{\alpha_2 \alpha_3 \alpha_4} \bigl( \tfrac{1}{9} u_1^{\alpha_1} u_2^{\alpha_2} u_1^{\alpha_3} u_2^{\alpha_4} + \tfrac{1}{9} u_2^{\alpha_1} u_1^{\alpha_2} u_2^{\alpha_3} u_1^{\alpha_4} \bigr) \\&
  + R^\mu{}_{\alpha_1 \nu \alpha_2} R^\nu{}_{\alpha_3 \mu \alpha_4} \bigl( \tfrac{1}{180} u_1^{\alpha_1} u_1^{\alpha_2} u_1^{\alpha_3} u_1^{\alpha_4} + \tfrac{1}{45} u_1^{\alpha_1} u_1^{\alpha_2} u_2^{\alpha_3} u_2^{\alpha_4} \\&\quad + \tfrac{1}{45} u_1^{\alpha_1} u_2^{\alpha_2} u_1^{\alpha_3} u_2^{\alpha_4} + \tfrac{49}{180} u_1^{\alpha_1} u_2^{\alpha_2} u_2^{\alpha_3} u_1^{\alpha_4} \\&\quad + \tfrac{1}{180} u_2^{\alpha_1} u_2^{\alpha_2} u_2^{\alpha_3} u_2^{\alpha_4} \bigr)
  + \dotsb
  \bigr].
\end{align*}
Naturally, the result is invariant under exchange of $u_1$ and $u_2$.

Next we calculate
\begin{align*}
  \Delta(z-w,z+\tilde{w})^\frac12
  &= \Delta(z+v_2-\parap{u_1}^{v_2},z+v_2+\parap{u_1}^{v_2})^\frac12 \\
  &=
  \begin{aligned}[t]
    \exp\bigl[&
      \tfrac{1}{3} R_{\alpha_1 \alpha_2} u_1^{\alpha_1} u_1^{\alpha_2}
      + \tfrac{1}{3} R_{\alpha_1 \alpha_2 ; \alpha_3} u_1^{\alpha_1} u_1^{\alpha_2} u_2^{\alpha_3} \\&
      + \bigl( \tfrac{2}{45} R^\mu{}_{\alpha_1 \nu \alpha_2} R^\nu{}_{\alpha_3 \mu \alpha_4} + \tfrac{1}{30} R_{\alpha_1 \alpha_2 ; \alpha_3 \alpha_4} \bigr) u_1^{\alpha_1} u_1^{\alpha_2} u_1^{\alpha_3} u_1^{\alpha_4} \\& + \tfrac{1}{6} R_{\alpha_1 \alpha_2 ; \alpha_3 \alpha_4} u_1^{\alpha_1} u_1^{\alpha_2} u_2^{\alpha_3} u_2^{\alpha_4}
      + \dotsb
    \bigr],
  \end{aligned}
  \\
  \Delta(z+w,z+\tilde{w})^\frac12
  &= \Delta(z+v_1-\parap{u_2}^{v_1},z+v_1+\parap{u_2}^{v_1})^\frac12 \\
  &=
  \begin{aligned}[t]
    \exp\bigl[&
      \tfrac{1}{3} R_{\alpha_1 \alpha_2} u_2^{\alpha_1} u_2^{\alpha_2}
      + \tfrac{1}{3} R_{\alpha_1 \alpha_2 ; \alpha_3} u_2^{\alpha_1} u_2^{\alpha_2} u_1^{\alpha_3} \\&
      + \bigl( \tfrac{2}{45} R^\mu{}_{\alpha_1 \nu \alpha_2} R^\nu{}_{\alpha_3 \mu \alpha_4} + \tfrac{1}{30} R_{\alpha_1 \alpha_2 ; \alpha_3 \alpha_4} \bigr) u_2^{\alpha_1} u_2^{\alpha_2} u_2^{\alpha_3} u_2^{\alpha_4} \\& + \tfrac{1}{6} R_{\alpha_1 \alpha_2 ; \alpha_3 \alpha_4} u_2^{\alpha_1} u_2^{\alpha_2} u_1^{\alpha_3} u_1^{\alpha_4}
      + \dotsb
    \bigr],
  \end{aligned}
  \\
  \Delta(z-w,z+w)^{-\frac12} &=
  \begin{aligned}[t]
    \exp\bigl[&
      - \tfrac{1}{3} R_{\alpha_1 \alpha_2} (u_1^{\alpha_1} - u_2^{\alpha_1}) (u_1^{\alpha_2} - u_2^{\alpha_2}) \\&
      - \bigl(\tfrac{2}{45} R^\mu{}_{\alpha_1 \nu \alpha_2} R^\nu{}_{\alpha_3 \mu \alpha_4} + \tfrac{1}{30} R_{\alpha_1 \alpha_2 ; \alpha_3 \alpha_4} \bigr) \\&\quad\times (u_1^{\alpha_1} - u_2^{\alpha_1}) (u_1^{\alpha_2} - u_2^{\alpha_2}) (u_1^{\alpha_3} - u_2^{\alpha_3}) (u_1^{\alpha_4} - u_2^{\alpha_4})
      \\&
      + \tfrac{1}{9} R_{\mu \alpha_1} R^\mu{}_{\alpha_2 \alpha_3 \alpha_4} (u_1^{\alpha_1} - u_2^{\alpha_1}) (u_1^{\alpha_2} u_2^{\alpha_3} u_1^{\alpha_4} - u_2^{\alpha_2} u_1^{\alpha_3} u_2^{\alpha_4})
    \bigr],
  \end{aligned}
  \\
  \Delta(z,z+\tilde{w})^{-1} &=
  \begin{aligned}[t]
    \exp\bigl[&
      -\tfrac{1}{6} R_{\alpha_1 \alpha_2} (u_1^{\alpha_1} + u_2^{\alpha_1}) (u_1^{\alpha_2} + u_2^{\alpha_2}) \\&
      - \tfrac{1}{12} R_{\alpha_1 \alpha_2 ; \alpha_3} (u_1^{\alpha_1} + u_2^{\alpha_1}) (u_1^{\alpha_2} + u_2^{\alpha_2}) (u_1^{\alpha_3} + u_2^{\alpha_3}) \\&
      - \bigl( \tfrac{1}{180} R^{\mu}{}_{\alpha_1 \nu \alpha_2} R^\nu{}_{\alpha_3 \mu \alpha_4} + \tfrac{1}{40} R_{\alpha_1 \alpha_2 ; \alpha_3 \alpha_4} \bigr) \\&\quad\times (u_1^{\alpha_1} + u_2^{\alpha_1}) (u_1^{\alpha_2} + u_2^{\alpha_2}) (u_1^{\alpha_3} + u_2^{\alpha_3}) (u_1^{\alpha_4} + u_2^{\alpha_4}) \\&
      - \tfrac{1}{18} R_{\mu \alpha_1} R^\mu{}_{\alpha_1 \alpha_2 \alpha_3} (u_1^{\alpha_1} + u_2^{\alpha_1}) (u_1^{\alpha_2} u_2^{\alpha_3} u_1^{\alpha_4} + u_2^{\alpha_2} u_1^{\alpha_3} u_2^{\alpha_4})
    \bigr].
  \end{aligned}
\end{align*}
where we used~\eqref{eq:zeta_expansion-1}, \eqref{eq:zeta_expansion-3} and~\eqref{eq:zeta_expansion-4} together with~\eqref{eq:v1-expansion}, \eqref{eq:v2-expansion}, \eqref{eq:w-expansion} and~\eqref{eq:tilde-w-expansion}.

All together, we obtain
\begin{align*}
  \Lambda(z,u_1,u_2)
  &= 2^{-d} \abs*{\frac{\partial(w,\tilde{w})}{\partial(u_1,u_2)}} \frac{\Delta(z-w,z+\tilde{w})^\frac12 \Delta(z+w,z+\tilde{w})^\frac12}{\Delta(z-w,z+w)^\frac12 \Delta(z,z+\tilde{w})} \\
  &=
  \begin{aligned}[t]
    \exp\bigl[&
    \tfrac{1}{3} R_{\alpha_1 \alpha_2} u_1^{\alpha_1} u_2^{\alpha_2}
    + R_{\alpha_1 \alpha_2 ; \alpha_3} \bigl( \tfrac{1}{6} u_1^{\alpha_1} u_1^{\alpha_2} u_2^{\alpha_3} + \tfrac{1}{6} u_2^{\alpha_1} u_2^{\alpha_2} u_1^{\alpha_3} \bigr) \\&
    + R_{\alpha_1 \alpha_2 ; \alpha_3  \alpha_4} \bigl(
      \tfrac{1}{120} u_1^{\alpha_1} u_1^{\alpha_2} u_1^{\alpha_3} u_2^{\alpha_4}
      + \tfrac{1}{120} u_1^{\alpha_1} u_1^{\alpha_2} u_2^{\alpha_3} u_1^{\alpha_4} \\&\quad
      + \tfrac{1}{60} u_1^{\alpha_1} u_2^{\alpha_2} u_1^{\alpha_3} u_1^{\alpha_4}
      - \tfrac{1}{60} u_1^{\alpha_1} u_2^{\alpha_2} u_2^{\alpha_3} u_1^{\alpha_4} \\&\quad
      - \tfrac{1}{60} u_1^{\alpha_1} u_2^{\alpha_2} u_1^{\alpha_3} u_2^{\alpha_4}
      + \tfrac{3}{40} u_1^{\alpha_1} u_1^{\alpha_2} u_2^{\alpha_3} u_2^{\alpha_4} \\&\quad
      + \tfrac{3}{40} u_2^{\alpha_1} u_2^{\alpha_2} u_1^{\alpha_3} u_1^{\alpha_4}
      + \tfrac{1}{120} u_2^{\alpha_1} u_2^{\alpha_2} u_2^{\alpha_3} u_1^{\alpha_4} \\&\quad
      + \tfrac{1}{120} u_2^{\alpha_1} u_2^{\alpha_2} u_1^{\alpha_3} u_2^{\alpha_4}
      + \tfrac{1}{60} u_1^{\alpha_1} u_2^{\alpha_2} u_2^{\alpha_3} u_2^{\alpha_4} \bigr) \\&
    + R_{\mu \alpha_1} R^{\mu}{}_{\alpha_2 \alpha_3 \alpha_4} \bigl(
      \tfrac{1}{18} u_1^{\alpha_1} u_1^{\alpha_2} u_2^{\alpha_3} u_1^{\alpha_4}
      - \tfrac{1}{18} u_1^{\alpha_1} u_2^{\alpha_2} u_1^{\alpha_3} u_2^{\alpha_4} \\&\quad
      - \tfrac{1}{18} u_2^{\alpha_1} u_1^{\alpha_2} u_2^{\alpha_3} u_1^{\alpha_4}
      + \tfrac{1}{18} u_2^{\alpha_1} u_2^{\alpha_2} u_1^{\alpha_3} u_2^{\alpha_4} \bigr) \\&
    + R^{\mu}{}_{\alpha_1 \nu \alpha_2} R^{\nu}{}_{\alpha_4 \mu \alpha_3} \bigl(
      \tfrac{7}{45} u_2^{\alpha_1} u_1^{\alpha_2} u_1^{\alpha_3} u_1^{\alpha_4}
      - \tfrac{7}{90} u_1^{\alpha_1} u_2^{\alpha_2} u_2^{\alpha_3} u_1^{\alpha_4} \\&\quad
      + \tfrac{31}{180} u_1^{\alpha_1} u_2^{\alpha_2} u_1^{\alpha_3} u_2^{\alpha_4}
      - \tfrac{7}{90} u_1^{\alpha_1} u_1^{\alpha_2} u_2^{\alpha_3} u_2^{\alpha_4} \\&\quad
      + \tfrac{7}{45} u_1^{\alpha_1} u_2^{\alpha_2} u_2^{\alpha_3} u_2^{\alpha_4} \bigr)
    + \dotsb
    \bigr].
  \end{aligned}
\end{align*}

\subsection{Expansions of the remaining factors}

Inserting in $\e^{2\i (w+u_1-u_2) \cdot p}$ the expansion of~$w$ in terms of~$u_1$ and $u_2$, as given by~\eqref{eq:w-expansion}, we obtain
\begin{equation*}\begin{split}
  \e^{2\i (w+u_1-u_2) \cdot p}
  = \smash{\exp\Bigl[}& \i p_\mu \smash{\Bigl(}
  R^\mu{}_{\alpha_1\alpha_2\alpha_3} \bigl(\tfrac13 u_1^{\alpha_1} u_2^{\alpha_2} u_1^{\alpha_3} - \tfrac13 u_2^{\alpha_1} u_1^{\alpha_2} u_2^{\alpha_3}\bigr) \\
  &\quad + R^\mu{}_{\alpha_1\alpha_2\alpha_3;\alpha_4} \bigl(-\tfrac13 u_2^{\alpha_1} u_1^{\alpha_2} u_2^{\alpha_3} u_2^{\alpha_4} + \tfrac13 u_1^{\alpha_1} u_2^{\alpha_2} u_1^{\alpha_3} u_1^{\alpha_4} \bigr) \\
  &\quad + R^\mu{}_{\alpha_1\alpha_2\alpha_3;\alpha_4\alpha_5} \bigl(- \tfrac{7}{60} u_2^{\alpha_1} u_1^{\alpha_2} u_2^{\alpha_3} u_2^{\alpha_4} u_1^{\alpha_5} + \tfrac{1}{60} u_1^{\alpha_1} u_2^{\alpha_2} u_1^{\alpha_3} u_2^{\alpha_4} u_2^{\alpha_5} \\&\qquad + \tfrac{1}{20} u_1^{\alpha_1} u_2^{\alpha_2} u_1^{\alpha_3} u_1^{\alpha_4} u_1^{\alpha_5} + \tfrac{7}{60} u_1^{\alpha_1} u_2^{\alpha_2} u_1^{\alpha_3} u_1^{\alpha_4} u_2^{\alpha_5} \\&\qquad - \tfrac{1}{60} u_2^{\alpha_1} u_1^{\alpha_2} u_2^{\alpha_3} u_1^{\alpha_4} u_1^{\alpha_5} - \tfrac{1}{20} u_2^{\alpha_1} u_1^{\alpha_2} u_2^{\alpha_3} u_2^{\alpha_4} u_2^{\alpha_5} \bigr) \\
  &\quad + R^\mu{}_{\alpha_1\beta\alpha_2} R^\beta{}_{\alpha_3\alpha_4\alpha_5} \bigl(\tfrac{7}{180} u_2^{\alpha_1} u_2^{\alpha_2} u_1^{\alpha_3} u_2^{\alpha_4} u_1^{\alpha_5} - \tfrac{11}{180} u_2^{\alpha_1} u_1^{\alpha_2} u_2^{\alpha_3} u_1^{\alpha_4} u_2^{\alpha_5} \\&\qquad - \tfrac{11}{45} u_1^{\alpha_1} u_2^{\alpha_2} u_2^{\alpha_3} u_1^{\alpha_4} u_2^{\alpha_5} + \tfrac{7}{180} u_1^{\alpha_1} u_1^{\alpha_2} u_1^{\alpha_3} u_2^{\alpha_4} u_1^{\alpha_5} \\&\qquad - \tfrac{7}{180} u_1^{\alpha_1} u_1^{\alpha_2} u_2^{\alpha_3} u_1^{\alpha_4} u_2^{\alpha_5} + \tfrac{11}{180} u_1^{\alpha_1} u_2^{\alpha_2} u_1^{\alpha_3} u_2^{\alpha_4} u_1^{\alpha_5} \\&\qquad + \tfrac{11}{45} u_2^{\alpha_1} u_1^{\alpha_2} u_1^{\alpha_3} u_2^{\alpha_4} u_1^{\alpha_5} - \tfrac{7}{180} u_2^{\alpha_1} u_2^{\alpha_2} u_2^{\alpha_3} u_1^{\alpha_4} u_2^{\alpha_5} \bigr) + \dotsb \smash{\Bigr) \Bigr]}
\end{split}\end{equation*}

Expanding $a(z+v_1,\parap{p}_{v_1})$ and $b(z+v_2,\parap{p}_{v_2})$ around~$z$ (using the horizontal derivative, see~\eqref{eq:taylor2}), and replacing~$v_1$ resp.~$v_2$ by their expansions with respect to~$u_1$ and~$u_2$, as given by~\eqref{eq:v1-expansion} and~\eqref{eq:v2-expansion}, we obtain
\begin{align*}
  \begin{split}
    a(z+v_1,\parap{p}_{v_1}) &=
    \bigl( a + a_{;\alpha} u_1^\alpha + \tfrac12 a_{;\alpha_1\alpha_2} u_1^{\alpha_1} u_1^{\alpha_2} + \tfrac16 a_{;\alpha_1\alpha_2\alpha_3} u_1^{\alpha_1} u_1^{\alpha_2} u_1^{\alpha_3} \\&\qquad + \tfrac{1}{24} a_{;\alpha_1\alpha_2\alpha_3\alpha_4} u_1^{\alpha_1} u_1^{\alpha_2} u_1^{\alpha_3} u_1^{\alpha_4} + \dotsb \bigr) \\
    &\quad + R^\beta{}_{\alpha_1\alpha_2\alpha_3} \bigl( \tfrac12 a_{;\beta} u_2^{\alpha_1} u_1^{\alpha_2} u_2^{\alpha_3} + \tfrac12 a_{;\beta\alpha_4} u_2^{\alpha_1} u_1^{\alpha_2} u_2^{\alpha_3} u_1^{\alpha_4} + \dotsb \bigr) \\
    &\quad + R^\beta{}_{\alpha_1\alpha_2\alpha_3;\alpha_4} \Bigl( a_{;\beta} \bigl(\tfrac{1}{12} u_2^{\alpha_1} u_1^{\alpha_2} u_2^{\alpha_3} u_2^{\alpha_4} - \tfrac{1}{24} u_1^{\alpha_1} u_2^{\alpha_2} u_1^{\alpha_3} u_2^{\alpha_4} \\&\qquad + \tfrac{5}{24} u_2^{\alpha_1} u_1^{\alpha_2} u_2^{\alpha_3} u_1^{\alpha_4} - \tfrac{1}{12} u_1^{\alpha_1} u_2^{\alpha_2} u_1^{\alpha_3} u_1^{\alpha_4}\bigr) + \dotsb \Bigr) + \dotsb
  \end{split}
  \\
  \begin{split}
    b(z+v_2,\parap{p}_{v_2}) &=
    \bigl(b + b_{;\alpha} u_2^\alpha + \tfrac12 b_{;\alpha_1\alpha_2} u_2^{\alpha_1} u_2^{\alpha_2} + \tfrac16 b_{;\alpha_1\alpha_2\alpha_3} u_2^{\alpha_1} u_2^{\alpha_2} u_2^{\alpha_3} \\&\qquad + \tfrac{1}{24} b_{;\alpha_1\alpha_2\alpha_3\alpha_4} u_2^{\alpha_1} u_2^{\alpha_2} u_2^{\alpha_3} u_2^{\alpha_4} + \dotsb \bigr) \\
    &\quad + R^\beta{}_{\alpha_1\alpha_2\alpha_3} \bigl(\tfrac12 b_{;\beta} u_1^{\alpha_1} u_2^{\alpha_2} u_1^{\alpha_3} + \tfrac12 b_{;\beta\alpha_4} u_1^{\alpha_1} u_2^{\alpha_2} u_1^{\alpha_3} u_2^{\alpha_4} + \dotsb \bigr) \\
    &\quad + R^\beta{}_{\alpha_1\alpha_2\alpha_3;\alpha_4} \Bigl( b_{;\beta} \bigl(\tfrac{1}{12} u_1^{\alpha_1} u_2^{\alpha_2} u_1^{\alpha_3} u_1^{\alpha_4} - \tfrac{1}{24} u_2^{\alpha_1} u_1^{\alpha_2} u_2^{\alpha_3} u_1^{\alpha_4} \\&\qquad + \tfrac{5}{24} u_1^{\alpha_1} u_2^{\alpha_2} u_1^{\alpha_3} u_2^{\alpha_4} - \tfrac{1}{12} u_2^{\alpha_1} u_1^{\alpha_2} u_2^{\alpha_3} u_2^{\alpha_4}\bigr) + \dotsb \Bigr) + \dotsb
  \end{split}
\end{align*}
where we recall that $a_{;\alpha}$, $b_{;\alpha}$ etc.\ denote the horizontal derivatives of the symbols $a$ and $b$.

\bigskip
\paragraph{Acknowledgments}
All authors gratefully acknowledge the financial support of the National Science Center (NCN), Poland, under the grant UMO-2014/15/B/ST1/00126.
The work of D.S.\ was partially supported by the NCN grant DEC-2015/16/S/ST1/00473 and the DFG project ``Solutions and Stability of the semiclassical Einstein equation – a phase space approach''.
We thank Wojciech Kami\'{n}ski for useful explanations.


\end{document}